\documentclass[12pt]{article}
  \usepackage[ ]{hyperref}
\usepackage{graphicx, amsmath,amsthm,amsopn,amsfonts,amssymb,dsfont, esint, oldgerm,yhmath,enumerate,bm}
\usepackage{color,yfonts}

 \usepackage{pdfpages}
 \usepackage{graphicx}
  \usepackage{fullpage}
\usepackage{color,yfonts}
\usepackage{array}

\newcommand\dblint[2]{\int_{#1}\int_Q #2 \ \mathrm{d}y\mathrm{d}x}
\newcommand\dblintn[2]{\int_{#1}\int_{NQ} #2 \mathrm{d}y\mathrm{d}x}
\newcommand\mv[1]{\langle #1 \rangle}
\newcommand\twoscale{\overset{2}{\rightharpoonup}}
\newcommand\ep{\varepsilon}

\newcommand\pd[2]{\tfrac{\partial #1}{\partial #2}}
\newcommand\integral[2]{\int_{#1}#2}

\renewenvironment{thebibliography}[1]{\begin{list}{\arabic{enumi}.}
{\usecounter{enumi}\setlength{\parsep}{0pt}}}{\end{list}}

\newtheorem{thm}{Theorem}[section]
\newtheorem{prop}{Proposition}[section]
\newtheorem{lem}{Lemma}[section]
\newtheorem{defn}{Definition}[section]
\newtheorem{rem}{Remark}[section]
\newtheorem{cor}{Corollary}[section]

\numberwithin{equation}{section}

\newcommand{\ZZ}{{\mathbb Z}}

\newcommand{\D}{{\mathcal D}}

\newcommand{\RR}{{\mathbb R}}
\newcommand{\CC}{{\mathbb C}}
\newcommand{\B}{{\mathcal B}}
\newcommand{\NN}{{\mathbb{N} }}

\def\H{{\mathcal H}}
\newcommand{\e}{\varepsilon}

\newcommand{\norm}[1]{\Vert#1\Vert}

\def\G{{\cal G}}

\def\A{{\cal A}}
\def\B{{\cal B}}

\def\D{{\cal D}}

\def\G{{\cal G}}
\def\H{{\cal H}}

\def\CC{\mathbb{C}}

\def\NN{\mathbb{N}}

\def\RR{\mathbb{R}}

\def\ZZ{\mathbb{Z}}
\def\RR{\mathbb{R}}

\def\a{\alpha}
\def\b{\beta}

\def\div{{\rm div }\hspace{2pt}}

\title{On band gaps in photonic crystal fibers}
\author{Shane Cooper, Ilia Kamotski and Valery Smyshlyaev}
\begin{document}
\maketitle
\begin{abstract} 
We consider the Maxwell's system for a periodic array of dielectric `fibers' embedded into a `matrix', with respective electric permittivities $\epsilon_0$ and $\epsilon_1$, which serves as a model for cladding in photonic crystal fibers (PCF). The interest is in describing admissible and forbidden (gap) pairs $(\omega,k)$ of frequencies $\omega$ and propagation constants $k$ along the fibers, for a Bloch wave solution on the cross-section. We show that, for ``pre-critical'' values of $k(\omega)$ i.e. those just below $\omega (\min\{\epsilon_0,\epsilon_1\}\mu)^{1/2}$ (where $\mu$ is the magnetic permeability assumed constant for simplicity), the coupling specific to the Maxwell's systems leads to a particular partially degenerating PDE system for the axial components of the electromagnetic field. Its asymptotic analysis allows to derive the limit spectral problem where the fields are constrained in one of the phases by Cauchy-Riemann type relations. We prove related spectral convergence. We finally give some examples, in particular of small size ``arrow'' fibers ($\epsilon_0>\epsilon_1$) where the existence of the gaps near appropriate ``micro-resonances'' is demonstrated by a further asymptotic analysis. 
 
\end{abstract}


\section{Problem formulation and main result}
\label{secPr}

We consider the Maxwell's system 
\begin{equation}
\label{pceq:pf1}
\begin{aligned}
\nabla \times \hat{H} &= - \, \epsilon\, \pd{\hat{E}}{t} \\
\nabla \times \hat{E} &=  \, \mu\, \pd{\hat{H}}{t}, 
\end{aligned}
\end{equation}
where the electric permittivity  $\epsilon$ and magnetic permeability $\mu$ adopt two different sets of constant values in the fibers  along the $x_3$ direction which are positioned periodically in the cross-sectional $(x_1,x_2)$-plane, and in the surrounding matrix. 

More precisely, let $Q : = [0,1)^2$ be the reference periodic cell, $Q_0$ be an open bounded  subset of $Q$, 
$\overline{Q_0}\subset Q$, with sufficiently smooth boundary $\Gamma$, and $Q_1 : = Q \backslash \overline{Q_0}$. 
Let $\chi_i$ denote the characteristic function of the $Q$-periodically extended $Q_i$, $i=0,1$. Then 
\begin{align*}
\epsilon(\hat y) & = \epsilon_0 \chi_0(y_1,y_2) + \epsilon_1 \chi_1(y_1,y_2), & \mu(\hat y) & = \mu_0 \chi_0(y_1,y_2) +\mu_1 \chi_1(y_1,y_2), & \hat y= (y_1,y_2,y_3) \in \RR^3, 
\end{align*} 
where $\epsilon_i$, $\mu_i$, $i=0,1$, are positive constants. 

A Bloch wave type solution to \eqref{pceq:pf1} is sought in the form 
\begin{equation} 
\hat{E}=e^{{\rm i}(\omega t + ky_3)}E(y_1,y_2), \ \ \ \  \hat{H}=e^{{\rm i}(\omega t + ky_3)}H(y_1,y_2).  
\label{ehy1y2}
\end{equation} 
The interest is in describing the pairs $(\omega, k)$ ($\omega\geq 0$ is the frequency and $k\geq 0$ is the 
``propagation constant'') for which there exists a non-trivial solution of the form \eqref{ehy1y2} with $E(y_1,y_2)$ and $H(y_1,y_2)$ quasi-periodic 
in $y=(y_1,y_2)$. 

Upon substituting \eqref{ehy1y2} into \eqref{pceq:pf1}, we find that $E=(E_1,E_2,E_3)$, $H=(H_1,H_2,H_3)$ necessarily satisfy  the following system of equations 
\begin{eqnarray}
H_{3,2} - {\rm i}kH_2 &=& -{\rm i} \omega \epsilon E_1 \label{3} \\
{\rm i}kH_1 - H_{3,1} &=& -{\rm i} \omega {\epsilon} E_2 \label{4} \\
H_{2,1} - H_{1,2} &=& -{\rm i} \omega {\epsilon} E_3  \label{5}
\end{eqnarray}
\begin{eqnarray}
E_{3,2} - {\rm i}kE_2 &=& {\rm i} \omega {\mu} H_1 \label{6} \\
{\rm i}kE_1 - E_{3,1} &=& {\rm i} \omega {\mu} H_2 \label{7} \\
E_{2,1} - E_{1,2} &=& {\rm i}\omega {\mu} H_3. \label{8} 
\end{eqnarray}
Resolving then \eqref{4} and \eqref{6} for $H_1$ and $E_2$, and \eqref{3} and \eqref{7} for $H_2$ and $E_1$,  
gives the following representations of the cross-sectional components $(E_1, E_2)$ and $(H_1, H_2)$ in terms 
of the ``axial'' components $E_3$ and $H_3$ 
\begin{eqnarray}
H_1 &=&  a^{-1} \big( {\rm i}k H_{3,1} - {\rm i} \omega {\epsilon} E_{3,2} \big), \label{9} \\ 
H_2 &=& {a}^{-1} \big( {\rm i}k H_{3,2} + {\rm i} \omega {\epsilon} E_{3,1}\big), \label{10} \\
E_1 &=& {a}^{-1} \big( {\rm i}k E_{3,1} + {\rm i} \omega {\mu} H_{3,2}\big), \label{11} \\
E_2 &=& {a}^{-1} \big( {\rm i}k E_{3,2} - {\rm i} \omega {\mu} H_{3,1}\big), \label{12}
\end{eqnarray}
assuming 
\begin{equation}
a(y) \,\, : = \,\, \omega^{2}\epsilon(y)\mu(y) \,\,- \,\, k^2 
\label{ay}
\end{equation} 
is not zero. Substituting then \eqref{9}--\eqref{12} into \eqref{5} and \eqref{8} reduces \eqref{3}--\eqref{8} to the following system for $E_3$ and $H_3$ only 
\begin{align}
\partial_{1}\left(\frac{{\rm i}k}{{a}} H_{3,2} \right) - \partial_{2}\left(\frac{{\rm i}k}{{a}} H_{3,1} \right) + \partial_{1}\left(\frac{{\rm i}\omega {\epsilon}}{{a}} E_{3,1} \right) + \partial_{2}\left(\frac{{\rm i}\omega {\epsilon}}{{a}} E_{3,2} \right) & = - {\rm i} \omega {\epsilon} E_3  \label{13} \\
\partial_{1}\left(\frac{{\rm i}k}{{a}} E_{3,2} \right) - \partial_{2}\left(\frac{{\rm i}k}{{a}} E_{3,1} \right) - \partial_{1}\left(\frac{{\rm i}\omega {\mu}}{{a}} H_{3,1} \right) - \partial_{2}\left(\frac{{\rm i}\omega {\mu}}{{a}} H_{3,2} \right) & =  {\rm i} \omega {\mu} H_3. \label{14}
\end{align}

Setting $u:=(E_3,H_3)=(u_1,u_2)$, multiplying \eqref{13} and \eqref{14} by smooth test functions ${\rm i}\omega\phi_1$ and $(-{\rm i}\omega\phi_2)$ respectively, adding up  and integrating over $y=(y_1,y_2)$ gives upon integration by parts the following equivalent weak formulation: 
\begin{equation}
\label{waveguideform}
\b(u,\phi) = \omega^2 \integral{\mathbb{R}^2}{ \rho u \cdot \overline{\phi}}, \,\,  \quad \forall \phi=(\phi_1,\phi_2) \in [C^{\infty}_{0}(\mathbb{R}^2)]^2.
\end{equation}
Here the bilinear form $\b$ 
is given by
\begin{multline}
\label{quadform}
\b(u,\phi) : = \int_{\mathbb{R}^2}\frac{\omega^2}{a}\Big( \epsilon \nabla  u_1 \cdot \overline{\nabla \phi_1} + \mu \nabla u_2 \cdot \overline{\nabla \phi_2} \Big) + \frac{k\omega}{a}\Big(  \left\{ u_1,\overline{\phi_2} \right\} - \left\{ u_2,\overline{\phi_1} \right\}  \Big),
\end{multline}
where $\left\{ v,w \right\} : = v_{,1}w_{,2} - w_{,1}v_{,2}$, 
$a = \omega^2\epsilon\mu - k^2$ as above,
\begin{align}
\label{electromagprop}
\epsilon & := \epsilon(y) = \epsilon_0 \chi_0(y) + \epsilon_1 \chi_1(y), & \mu & := \mu(y)= \mu_0 \chi_0(y) +\mu_1 \chi_1(y), & y \in \RR^2, 
\end{align}
and  
$$
\rho(y)\, :=\, \left(\begin{matrix} \epsilon(y) & 0 \\ 0 & \mu(y) \end{matrix}\right) \,=\,\chi_0(y) \left( \begin{matrix} \epsilon_0 & 0 \\ 0 & \mu_0 \end{matrix} \right) + \chi_1(y) \left( \begin{matrix} \epsilon_1 & 0 \\ 0 & \mu_1 \end{matrix} \right).
$$
Notice that the form \eqref{quadform} is symmetric, i.e. $\beta(\phi,u)=\overline{\beta(u,\phi)}$, as follows by direct inspection. 
Noticing further that, for any $\gamma>0$, 
\[
\left\vert\left\{ u_1,\overline{u_2} \right\} - \left\{ u_2,\overline{u_1} \right\}\right\vert\,\leq 
\gamma\vert\nabla u_1\vert^2+\gamma^{-1}\vert \nabla u_2\vert^2, 
\]
it is readily seen that if  $\omega,k$ satisfy 
\begin{equation}
\label{positivitycond}
k^2 < \omega^2\text{min}\{\epsilon(y)\mu(y)\}\,=\,\text{min} 
\{\epsilon_0\mu_0, \epsilon_1\mu_1\},
\end{equation}
then $\b$ is coercive on $[H^1(\RR^2)]^2$; namely, there exists a constant $\nu >0$ such that, 
\begin{equation}
\label{positivitycond2}
\b(u,u) \ge \nu \norm{u}_{[H^1(\mathbb{R}^2)]^2}^2, \quad \forall u \in [H^1(\mathbb{R}^2)]^2.
\end{equation}

Henceforth,  we shall consider a `non-magnetic' photonic crystal fiber, i.e. we set the magnetic permeability $\mu_0 =\mu_1=\mu$ a constant, and we assume that $\epsilon_0 > \epsilon_1>0$ in \eqref{electromagprop}. We shall consider the set of pairs $(\omega^2, k^2)$ for which the problem \eqref{waveguideform} admits a non-trivial solution when $k$ approaches from below the critical line  
$k = \omega \left(\epsilon_1\mu\right)^{1/2}$. (This line corresponds to the dispersion relation in the matrix - the wavenumber for a plane wave - which appears exactly where $\b$ loses its coercivity.) 
More precisely, for a fixed small parameter $\ep >0$,  we can say that \eqref{waveguideform} admits a non-trivial solution for the pair $(\omega^2, \omega^2\mu(\epsilon_1 - \ep^2))$   if $\lambda:=\omega^2$ belongs to the spectrum $\sigma_\ep$ of the self-adjoint operator $\B^\ep$ 
generated by the bilinear form $\b^\ep : [H^1(\RR^2)]^2 \times [H^1(\RR^2)]^2 \rightarrow \CC$ given by \eqref{quadform}. This form can be conveniently represented as 
\begin{equation}
\label{19}
\b^\ep(u,v) = \integral{\mathbb{R}^2}{ A^\ep \nabla u \cdot \overline{\nabla v}}  = \integral{\mathbb{R}^2}{ A^\ep_{ijpq}  u_{p,q}\overline{v}_{i,j} }.
\end{equation}
Here $A^\ep(y) = \ep^{-2}\chi_1(y) A^\ep_{1} + \left(\epsilon_0 - \epsilon_1 + \ep^2\right)^{-1} \chi_0(y)A^\ep_{0}$ is easily found by rearranging \eqref{quadform} upon substituting $k^2 = \omega^2 \mu ( \epsilon_1 - \ep^2)$; see \eqref{tensor2} below. 
Then the domain and the range of the self-adjoint operator $\B^\ep$ consist of all $u\in [H^1(\RR^2)]^2$ and $f\in [H^1(\RR^2)]^2$ respectively, 
such that ($\B^\ep u = f$) 
$$
\b^\ep(u,\phi) = \int_{\RR^2} \rho f \cdot \overline{\phi}, \quad \forall \phi \in [C^\infty_0(\RR^2)]^2.
$$

Notice that, for a fixed $\varepsilon>0$, the operator $\B^\ep$ and hence its spectrum $\sigma_\ep$ depend on the spectral parameter 
$\lambda=\omega^2$, cf. \eqref{quadform}, \eqref{ay}, hence one generally deals with an operator pencil. We will see however that, 
as $\ep \to 0$, it asymptotically bahaves as a conventional spectral problem (although for a ``partially degenerating'' operator). 
 
So, for $\ep$ converging to zero, we aim to characterise the (Floquet-Bloch) spectrum $\sigma_\ep$. This brings us to our main result
\begin{thm}
\label{mainthm1}
Let $\sigma_\ep$ be the spectrum of the operator $\B^\ep$ generated by the bilinear form $\b^\ep$, see \eqref{19}. Then the set $\sigma_\ep$ converges to the set
$$
\sigma_0 =  \bigcup_{\theta \in [0,1)^2} \sigma (\theta)
$$
as $\ep \rightarrow 0$ in the sense of Hausdorff. 
Here $\sigma(\theta)$ is the set of $\lambda_k(\theta)$, $k \in \NN$, such that there exists a non-zero solution $u \in V(\theta)$ to 
\begin{multline}
\int_{Q}   \nabla u_1 \cdot \overline{\nabla \phi_1} + \gamma^{-1} \Big( \div{ u} \cdot \overline{ \div{\phi}} + \div{ u^\perp} \cdot \overline{ \div{ \phi^\perp}} \Big)
= \lambda_k(\theta) \int_{Q} \epsilon_1 \left( u \cdot \overline{\phi}  + \gamma \chi_0 u_1 \overline{ \phi_1} \right), \\ \qquad \forall \phi \in V(\theta),  \label{122} 
\end{multline}
where  $\gamma : = {\epsilon_0}/{\epsilon_1} - 1>0$ and the unitary mapping ${}^\perp: \RR^2 \rightarrow \RR^2$ is defined by $a^\perp = (-a_2, a_1)$; 
\begin{equation}
\label{emspaceVn}
V(\theta) = \left\{ v \in [H^1_{\theta}(Q)]^2 : \text{ $\div{v} = 0$ and  $\div{v^{\perp}} = 0$ in  $Q_1$}\right\}, 
\end{equation}
where{\footnote{To be precise, $H^1_{\theta}(Q)$ is taken to be the closurer of smooth $\theta$-quasiperiodic  functions in the standard $H^1(Q)$ norm.}} $H^1_{\theta}(Q)$ is the space of  $H^1(Q)$ functions  that are $\theta$-quasiperiodic:  $u(y)=e^{i 2 \pi  \theta \cdot y} v(y)$ for some $v \in H^1_{\#}(Q)$.
\end{thm}
\begin{rem}
We say that sets $\sigma_\ep$ converge to a set $\sigma_0$ as $\ep \rightarrow 0$ in the sense of Hausdorff when the following two conditions hold:
\begin{enumerate}[(i)]
\item For every $\lambda \in \sigma_0$ there exists $\lambda_\ep \in \sigma_\ep$ such that  $\lambda_\ep \rightarrow \lambda$ as $\ep \rightarrow 0$.
\item Suppose $\lambda_\ep \in \sigma_\ep$ such that $\lambda_\ep$ converges to some $\lambda$ as $\ep \rightarrow 0$, then  $\lambda \in \sigma_0$.
\end{enumerate}
Symbolically, we represent this convergence as follows
$$
\lim_{\ep \rightarrow 0} \sigma_\ep = \sigma_0.
$$
\end{rem}
\begin{rem}
\label{rem1}
For fixed $\theta \in [0, 1)^2$, the eigenvalues $\lambda_k(\theta)$ are ordered, and  repeated according to their multiplicity, as follows:
$$
0 \le \lambda_1 (\theta) \le \lambda_2 (\theta) \le \lambda_3 (\theta) \le \ldots, \qquad \lambda_k(\theta) \rightarrow \infty \ \text{as} \ k \rightarrow \infty.
$$
Moreover, for any fixed $k$ the eigenvalues $\lambda_k(\theta)$ are continuous with respect to $\theta$ (see Section \ref{secblochlim}, Lemma \ref{apb.1}), and one has
$$
\sigma_0 = \bigcup_{i=1}^{\infty} [\min_{\theta} \lambda_i (\theta),\max_{\theta} \lambda_i (\theta)].
$$
\end{rem}
\begin{cor}
\label{equivformoflimspec}
The set $\sigma_0$ also admits the following equivalent representation 
$$
\sigma_0 = \bigcup_{N \in \NN^2} \sigma\left( \A_N \right),
$$
where $\sigma\left( \A_N \right)$
 is the set of $\lambda_k^N(\theta)$, $k \in \NN$, such that there exists a non-zero solution $u \in V_N$ to 
\begin{multline}
\int_{NQ}   \nabla u_1 \cdot \overline{\nabla \phi_1} + \gamma^{-1} \Big( \div{ u} \cdot \overline{ \div{\phi}} + \div{ u^\perp} \cdot \overline{ \div{ \phi^\perp}} \Big)
= \lambda_k^N \int_{NQ} \epsilon_1 \left( u \cdot \overline{\phi}  + \gamma \chi_0 u_1 \overline{ \phi_1} \right), \\ \qquad \forall \phi \in V_N,\end{multline}
where
$$ 
\quad V_N : = \left\{ v \in [H^1_{\#}(NQ)]^2 : \ \text{ $\div{v} = 0$ and  $\div {v^{\perp}}$ = 0 in  $F_1 \cap NQ$}\right\}.
$$
Here $NQ : = [0, N_1) \times [0, N_2)$ for a given  multi-index $N \in \NN^2$ is a ``multi-cell'' of size 
$N_1\times N_2$, and $F_1$ is the $Q$-periodic extension of the set $Q_1$ to the whole space $\RR^2$.
\end{cor}
\begin{proof}
By Theorem \ref{mainthm1}, inclusions \eqref{limspec01} and \eqref{limspec02}, one has
$$
\bigcup_{\theta \in [0,1)^2} \sigma(\theta) = \lim_{\ep \rightarrow 0} \sigma_\ep \supset \overline{\bigcup_{N \in \NN^2} \sigma\left( \A_N\right) } \supset\bigcup_{\theta \in [0,1)^2} \sigma(\theta).
$$
\end{proof}
The `limit spectrum' $\sigma_0$ has a gap if two adjacent bands do not overlap: i.e. if for some $i \in \mathbb{N}$, one has 
$$\max_{\theta} \lambda_i (\theta) < \min_{\theta} \lambda_{i+1} (\theta).$$
An immediate consequence of Theorem \ref{mainthm1} is the following result
\begin{cor}
Assume that $I$ is a gap of $\sigma_0$. Then, there exists a set 
$\G \subset \left\{ (\omega^2, k^2)\in\RR^2: \, k^2<\omega^2\mu\epsilon_1\right\} $ such that for any pair $(\omega^2, k^2) \in \G$ equation \eqref{13}-\eqref{14} admits no non-trivial solution. Furthermore, one has  $ \partial{  \G} \cap \{ (\omega^2, \omega^2 \mu\epsilon_1)\ \vert \ \omega^2 \in I \}  \neq \emptyset$.
\end{cor}
\begin{proof}
Suppose $\sigma_0$ has a gap, i.e. there exists an $i \in \NN$ such that $\max_\theta \lambda_i(\theta) < \min_\theta \lambda_{i+1}(\theta) $. Let us fix $a,b$ such that $\max_\theta \lambda_i(\theta) < a < b < \min_\theta \lambda_{i+1 }(\theta)$.

We now argue that there exists a $\ep_0 > 0$ such that  the set denoted by 
$$
\G = \bigcup_{\substack{\ep \\ 0<\ep < \ep_0}}  \left\{ \big(\omega^2, \omega^2\mu(\epsilon_1 - \ep^2)\big)\ \big\vert\  \omega^2 \in [a,b] \right\}
$$
satisfies the statements of the corollary. For if not, there exists a sequence $\ep \rightarrow 0$ for which each set $ \left\{ \big(\omega^2, \omega^2\mu(\epsilon_1 - \ep^2)\big)\ \big\vert\  \omega^2 \in [a,b] \right\}$ contains a pair $(\omega^2_\ep , \omega^2_\ep \mu( \epsilon_1 - \ep^2))$ such that \eqref{13}-\eqref{14} admits a non-trivial solution, i.e. $\lambda_\ep = \omega^2_\ep$ belongs to $\sigma_\ep$. Since $ \lambda_\ep \in [a,b] $, it has a subsequence which converges to some $\lambda \in [a,b]$. By Theorem \ref{mainthm1}, we conclude that $\lambda \in \sigma_0$ which contradicts the fact $ [a,b]$ is a gap in the spectrum $\sigma_0$.
\end{proof}
In Section \ref{examples} we prove the existence of gaps in the limit operator $\sigma_0$ for several subclasses of two-dimensional photonic crystal.
\section{Limit Bloch operators and some properties}
\label{secblochlim}
In this section we shall study some properties of the family of self-adjoint operators $\A(\theta)$, $\theta \in [0,1)^2$, associated with Theorem \ref{mainthm1}. In particular, we shall establish several important  properties concerning the continuity of these operators, their domains and their spectra with respect to the quasi-periodicity parameter $\theta$.

Suppose $\theta \in [0,1)^2$ and let $V(\theta)$ be defined by \eqref{emspaceVn}. 
\begin{rem}
$v \in V(\theta)$ is `equivalent' to $v$ solving the conjugate Cauchy-Riemann equations in $Q_1$: for $z= y_1 +i y_2$ if we define a complex valued $\theta$-quasiperiodic function $F(z) = v_1(y) - i v_2(y)$ then $F$ solves the Cauchy-Riemann equations in $Q_1$ if, and only if, $v \in V(\theta)$.
\end{rem}
Note that $V(\theta)$ is a closed subspace of $[H^1_\theta(Q)]^2$ and therefore a Hilbert space when equipped with the following equivalent $[H^1_\theta(Q)]^2$ norm
\begin{equation}
\label{normofV}
\norm{v}_{V(\theta)}^2  : = \int_Q \left\vert  v \right\vert^2 +  \int_Q \left\vert \div v \right\vert^2  +  \int_Q \vert \div v^\perp \vert^2.
\end{equation}
\begin{rem}
\label{remequ}
The equivalence of the norm \eqref{normofV} follows from standard Poincar\'{e} type inequalities and the following equality, which is a simple consequence of integration by parts and the quasiperiodicity of the functions in question, 
\begin{equation*}
\int_Q \vert \nabla v \vert^2 = \int_Q \vert \div v \vert^2 + \int_Q \vert \div v^\perp \vert^2 , \qquad \forall v \in [H^1_{\theta}(Q)]^2.
\end{equation*}
\end{rem}
Denoting by $\H_\theta$ the closure of $V(\theta)$ in\footnote{$L^2_\rho(Q)$ denotes the space of Lesbesgue measure functions $f$ such that $\rho f \cdot f$ is integrable for the matrix $\rho$ given by \eqref{rho}. Clearly, $L^2_\rho(Q)$ is equivalent to $[L^2(Q)]^2$. } $L^2_{\rho}(Q)$ we introduce the non-negative quadratic form $\a_\theta : \H_\theta \times \H_\theta \rightarrow \mathbb{R}$, defined by, cf. \eqref{122}, 
\begin{equation}
\label{biformtheta}
 \a_\theta(u,v) : = \int_{Q} \nabla u_1 \cdot \overline{\nabla v_1} + \gamma^{-1} \left( \div{ u} \cdot \overline{  \div v} + \div  u^\perp \cdot \overline{ \div v^\perp} \right).
\end{equation}
It is clear that $\a_\theta$ is closed in $\H_\theta$ due to the norm \eqref{normofV}. Therefore, $\a_\theta$ generates a unique self-adjoint operator $\A(\theta): \D(\theta) \rightarrow [L^2_\rho(Q)]^2$ whose domain $\D(\theta)$ is a dense subset of $V(\theta)$ and whose action $\A(\theta) u =f$ associates $f$ to the unique solution $u \in V(\theta)$ of
\begin{equation}
\a_\theta(u,\phi) 
 = \int_{Q} \epsilon_1 \left( f \cdot \overline{\phi}  + \gamma \chi_0 f_1 \overline{ \phi_1} \right)
=:\int_Q \epsilon_1\rho f\cdot \overline{\phi},  
\qquad \forall \phi \in V(\theta). 
\label{thetaprob}
\end{equation}
 As $H^1(Q)$ is compactly embedded into $L^2(Q)$ we find that the resolvent of $\A(\theta)$ is  compact and, in particular, the spectrum $\sigma(\theta)$ of  $\A(\theta)$ is discrete: it consists of countably many isolated eigenvalues of finite multiplicity which converge to infinity. We order the eigenvalues, allowing for  multiplicity, as follows:
$$
0 \le \lambda_1 (\theta) \le \lambda_2 (\theta) \le \lambda_3 (\theta) \le \ldots, \qquad \lambda_k(\theta) \rightarrow \infty \ \text{as} \ k \rightarrow \infty.
$$

An important distinct feature of the operators $\A(\theta)$ (as opposed to the usual Floquet-Bloch operators concerned with periodic elliptic PDEs in the whole space) is the non-trivial dependence on $\theta$ of the operator form domain $V(\theta)$. The statement of continuity of $\A(\theta)$ is therefore not a simple consequence of the Bloch-wave representation of functions belonging to $H^1_\theta(Q)$, and relies on establishing that the underlying space $V(\theta)$ is continuous with respect to $\theta$, see Lemma \ref{lem:spcom1}. We shall show that the continuity of $V(\theta)$ leads to the operators $\A(\theta)$ being continuous with respect to $\theta$ in the norm-resolvent sense, which in turn implies the continuity of $\lambda_k(\theta)$, see Lemma \ref{apb.1}.

In the rest of this section we shall introduce and prove the aforementioned results.
\begin{lem}[Continuity of $V(\theta)$]
\label{lem:spcom1}
Let $\theta^\ep \in [0,1)^2$ be such that $\theta^\ep \rightarrow \theta$ for some $\theta \in [0,1)^2$. Then, for each $\phi \in V(\theta)$ there exists $\phi^\ep \in V(\theta^\ep)$ such that $\phi^\ep \rightarrow \phi$ strongly in $[H^1(Q)]^2$ as $\ep \rightarrow 0.$
\end{lem}

To prove this lemma we shall use the following characterisation of the spaces $V(\theta)$.

\begin{lem}
\label{lem:spcom2} $ $ \\
\noindent (i). Suppose $\theta \in [0, 1)^2\backslash\{0\}$, then $\phi \in V(\theta)$ if, and only if,
$$
\phi = \nabla a + \nabla^{\perp}  b 
$$
for some $a,b \in H^2_{\theta}(Q)$ with
\begin{alignat}{3}
\Delta a  = f_1, & \hspace{3cm} & \Delta b  = f_2,
\end{alignat}
where $f_1, f_2 \in L^2(Q)$ with $\mathrm{supp}{f_1}$ and $\mathrm{supp}{f_2}$ contained in $\overline{Q_0}$. Here, for a given scalar function $f$, we denote by $\nabla^{\perp} f$ the vector $(-f_{,2}, f_{,1})$.

\noindent (ii). Suppose $\theta = 0$. Then $\phi \in V(0)$ if, and only if,
$$
\phi = c + \nabla a + \nabla^{\perp} b 
$$
for some constant $c \in \CC^2$ and $a,b \in H^2_{\#}(Q)$ with
\begin{alignat}{3}
\Delta a  = f_1, & \hspace{3cm} & \Delta b  = f_2,
\end{alignat}
where $f_1, f_2 \in L^2(Q)$, $\mv{f_1}=\mv{f_2}=0$,\footnote{Henceforth, $\mv{f}$ denotes the mean-value over $Q$ of an integrable function $f$, i.e. $\mv{f} := \int_Q f$.} with  $\mathrm{supp}{f_1}$ and $\mathrm{supp}{f_2}$ contained in $\overline{Q_0}$. 
\end{lem}

\begin{proof} In both cases (i) and (ii) the necessity condition is easy to demonstrate. Let us now show the sufficiency condition. First, we shall consider case (i). Suppose, $\theta \in [0,1)^2\backslash\{0\}$, and let us fix $\phi \in V(\theta)$ and set $f_1 : = \div \phi$, $f_2 : = \div \phi^\perp$. It is clear that $f_1,f_2 \in L^2(Q)$ and that their supports are contained in $\overline{Q_0}$. Now, let us introduce $a,b \in H^2_\theta (Q)$ the unique solutions of 
\begin{alignat*}{3}
\Delta a  = f_1, & \hspace{3cm} & \Delta b  = f_2.
\end{alignat*}
We claim that $\phi = \nabla a + \nabla^{\perp} b$. Indeed, this can be seen by noting that $w : = \phi - \nabla a - \nabla^{\perp} b$ belongs to $[H^1_{\theta}(Q)]^2$ with $\Delta w = 0$, which implies $w \equiv 0$. \\

For case (ii), fix $\phi \in V$; then it is clear that $\phi = c + \phi_0$ for some $\phi_0 \in V$ such that $\mv{\phi_0} = 0$. Now we repeat the above argument, that is let $a,b \in H^2_{\#}(Q)$ be the unique solutions of 
\begin{alignat*}{3}
\Delta a  = \div \phi_0, \quad \mv{a}=0, & \hspace{3cm} & \Delta b  = \div \phi_0^\perp, \quad \mv{b}=0.
\end{alignat*}
As above, we show that $\phi_0 = \nabla a + \nabla^{\perp} b$.
\end{proof}
Additionally, we shall use in the proof of Lemma \ref{lem:spcom1} the following regularity results.
\begin{lem} 
\label{lem:spcom3.1}
\hspace{0pt}

\begin{enumerate}[(i)]
\item{Suppose $\theta \in (0,1)^2$, and let $u \in H^1_{\theta}(Q)$ be the unique solution of 
\begin{equation}
\label{eq:spcom4.4}
- \Delta u = f,
\end{equation}
for some $f \in L^2(Q)$. Then,  $u \in H^2_{\theta}(Q)$ and 
\begin{equation}
\label{eq:spcom5}
\norm{u}_{H^2(Q)} \le  \left( 1 + \frac{1}{\vert \theta \vert^2} \right) \norm{f}_{L^2(Q)}.
\end{equation}}
\item{Suppose $\theta \in [0,1)^2$ and let $ u \in H^1_{\#}(Q)$, $\mv{u}=0$ be the unique solution of 
\begin{equation}
\label{eq:spcom8.1}
\Delta u + 4\pi {\rm i} \theta \cdot \nabla u - 4\pi^2\vert \theta \vert^2 u = f,
\end{equation}
for some $f \in L^2(Q)$ such that $\mv{f} = 0$. Then, $u \in H^2_{\#}(Q)$ and there exists a constant $C>0$, independent of $\theta$, such that
$$
\norm{u}_{H^2(Q)} \le C \norm{f}_{L^2(Q)}.
$$}
\end{enumerate}
\end{lem}
\begin{proof} \hspace{0pt}

{\bf Proof of (i):} The eigenfunctions $w_z(\theta;y) = e^{2\pi{\rm i}(\theta + z) \cdot y}$ of the $\theta$-quasiperiodic Laplacian form an orthonormal basis of $L^2(Q)$. By decomposing $u$ and $f$ in terms of this basis, we have
\begin{alignat*}{3}
u(y)  = \sum_{z \in \mathbb{Z}^2} a_z  e^{ 2\pi{\rm i}(\theta + z) \cdot y}, & \hspace{3cm} & f(z)  =  \sum_{z \in \ZZ^2} b_z  e^{ 2\pi{\rm i}(\theta + z) \cdot y},
\end{alignat*}
Now \eqref{eq:spcom4.4} tells us $a_z = \tfrac{b_z}{\lambda_z(\theta)}$, where $\lambda_z(\theta) = 4\pi^2 \vert \theta + z \vert^2$. Since
$$
\norm{u}_{H^2(Q)}^2 = \sum_{z \in \mathbb{Z}^2} \left( 1 + 4\pi^2\vert \theta + z \vert^2 \right)^2 \vert a_z \vert^2 = \sum_{z \in \mathbb{Z}^2} \left( \frac{1 + 4\pi^2\vert \theta + z \vert^2}{4\pi^2\vert \theta +  z \vert^2} \right)^2 \vert b_z \vert^2.
$$ 
For $\vert z \vert \ge 0$ we see $\tfrac{1 + 4 \pi^2\vert \theta + z \vert^2}{4\pi^2\vert \theta + z \vert^2} = \tfrac{1 }{4\pi^2\vert \theta + z \vert^2} + 1 \le \tfrac{1 }{\vert \theta \vert^2} + 1 $. Hence,  
$$
\norm{u}_{H^2(Q)} \le  \left( \frac{1 + \vert \theta \vert^2}{\vert \theta \vert^2} \right) \norm{f}_{L^2(Q)}.
$$ 

{\bf Proof of (ii):} The eigenfunctions $w_z(y ) = e^{2\pi {\rm i } z \cdot y}$ of the periodic Laplacian form an orthonormal basis in $L^2(Q)$. Decomposing $u$ and $f$ in terms of this basis we have
\begin{align*}
u(y) & = \sum_{z \in \mathbb{Z}^2} a_z  e^{ 2\pi {\rm i}z \cdot y} & f(z) & =  \sum_{z \in \ZZ^2} b_z  e^{2\pi {\rm i}z \cdot y},
\end{align*}
By the assumptions $\mv{f}=0$ and $\mv{u}=0$, we have $b_0 = 0$ and $a_0 =0$ respectively. Now \eqref{eq:spcom8.1} tells us that for $z \neq 0$, $a_z = - \tfrac{b_z}{4\pi^2\vert z + \theta \vert^2}$. Since $\theta \in [ 0, 1)^2$, one has
\begin{flalign*}
\norm{u}_{H^2(Q)}^2 &  = \sum_{\substack{z \in \mathbb{Z}^2 \\ z \neq 0}} \left( 1 + 4\pi^2 \vert   z \vert^2 \right)^2 \vert a_z \vert^2 = \sum_{\substack{z \in \mathbb{Z}^2 \\ z \neq 0}} \left( \frac{1 + 4\pi^2\vert  z \vert^2}{4\pi^2\vert \theta + z \vert^2} \right)^2 \vert b_z \vert^2 \\ & \le \sum_{\substack{z \in \mathbb{Z}^2 \\ z \neq 0}}  \left( \frac{1 + 4\pi^2 \vert  z \vert^2}{ 4\pi^2 \vert  z \vert^2} \right)^2 \vert b_z \vert^2\le 4 \sum_{\substack{z \in \mathbb{Z}^2 \\ z \neq 0}}  \vert b_z \vert^2,
\end{flalign*}
i.e.
$$
\norm{u}_{H^2(Q)}^2 \le 4  \norm{f}_{L^2(Q)}^2.
$$ 
\end{proof}

\begin{proof}[Proof of Lemma \eqref{lem:spcom1}.] Consider a sequence $\theta^\ep \subset [ 0 , 1)^2$ such that $\theta^\ep \rightarrow \theta$ as $\ep \rightarrow 0$. There are two separate cases to consider, the cases $\theta \in (0,1)^2$ and $\theta = 0$. \\

{\bf Case 1:} $\theta \in (0,1)^2$. Let us assume without loss of generality that $\theta^\ep \in (0,1)^2$. For fixed $\phi \in V(\theta)$ we know, by Lemma \ref{lem:spcom2}, that $\phi = \nabla a + \nabla^{\perp} b$ for some $a,b \in H^2_{\theta}(Q)$ where
\begin{alignat*}{3}
\Delta a = f_1, &\hspace{3cm}& \Delta b = f_2,
\end{alignat*}
for some $f_1,f_2 \in L^2(Q)$. We shall now construct the desired $\phi^\ep \in V(\theta^\ep)$ as follows: Set $\phi^\ep : = \nabla a_\ep + \nabla^{\perp} b_\ep$ where $a_\ep,b_\ep \in H^2_{\theta^\ep}(Q)$ solve
\begin{alignat*}{3}
\Delta a_\ep = f_1, &\hspace{3cm}& \Delta b_\ep = f_2.
\end{alignat*}
Notice, by Lemma \ref{lem:spcom2}, that $\phi^\ep \in V(\theta^\ep)$. It remains to show $\phi^\ep \rightarrow \phi$ strongly in $[H^1(Q)]^2$. To this end, it is sufficient to show that $a_\ep \rightarrow a$ and $b_\ep \rightarrow b$ strongly in $H^2(Q)$ as $\ep \rightarrow 0$. Let us show that $a_\ep \rightarrow a$ as $\ep \rightarrow 0$. By defining $\widetilde{a}_\ep (y) : = e^{-  2\pi{\rm i}( \theta^\ep - \theta)\cdot y} a_\ep(y)$, one notices that $\widetilde{a}_\ep \in H^2_{\theta}(Q)$ and uniquely solves 
\begin{equation}
\label{eq:spcom7}
\Delta \widetilde{a}_\ep  = f_\ep,
\end{equation} 
where $f_\ep : = e^{- 2\pi{\rm i} (\theta^\ep - \theta ) \cdot y}  f_1 - 4\pi {\rm i} \left( \theta^\ep - \theta \right) \cdot \nabla a_\ep - 4\pi^2 \vert \theta^\ep - \theta \vert^2 a_\ep $. In particular, $\widetilde{a}_\ep - a \in H^2_{\theta}(Q)$ with $\Delta (\widetilde{a}_\ep - a) = f_\ep - f$, and therefore by Lemma \ref{lem:spcom3.1}(i) one has
\begin{equation}
\label{spcomcorrect1}
\norm{\widetilde{a}_\ep - a}_{H^2(Q)}^2  \le C \norm{f_\ep - f_1}_{L^2(Q)}^2.
\end{equation}
Furthermore, since $\Delta a_\ep = f_1$,  by an application of Lemma \ref{lem:spcom3.1}(i) one notices 
$$
\norm{a_\ep}_{H^2(Q)}^2  \le  \left( 1 + \frac{1}{\vert \theta^\ep \vert^2} \right)^2 \norm{f_1}_{L^2(Q)}^2 \le C \norm{f_1}_{L^2(Q)}^2,
$$
where $C$ is independent of $\ep$ (here we used the fact that $\theta^\ep \rightarrow \theta^0 \neq 0 $ as $\ep \rightarrow 0$). Therefore, one has
\begin{multline*}
\norm{f_\ep - f_1}_{L^2(Q)}^2 \le \norm{ \left( e^{2\pi {\rm i} (\theta^\ep - \theta ) \cdot y} -1\right)  f_1}_{L^2(Q)}^2 +   \norm{  4\pi {\rm i}\left( \theta^\ep - \theta \right) \cdot \nabla a_\ep}_{L^2(Q)}^2  \\ + \norm{ 4 \pi^2 \vert \theta^\ep - \theta \vert^2 a_\ep}_{L^2(Q)}^2  \longrightarrow 0,
\end{multline*}
as $\ep \rightarrow 0$. Hence, by \eqref{spcomcorrect1}, $\tilde{a}_\ep \rightarrow a$ strongly in $H^2(Q)$ as $\ep \rightarrow 0$. Now, we can show $a_\ep \rightarrow a$ strongly in $H^2(Q)$ by noting that
\begin{flalign*}
\norm{a_\ep - a}_{H^2(Q)}^2 & = \norm{ e^{ 2\pi{\rm i} ( \theta^\ep - \theta) \cdot y} \widetilde{a}_\ep - a}_{H^2(Q)}^2\\ 
&  \le \norm{e^{ 2\pi{\rm i} ( \theta^\ep - \theta) \cdot y} ( \widetilde{a}_\ep - a ) }_{H^2(Q)}^2 + \norm{(e^{ 2\pi{\rm i} ( \theta^\ep - \theta) \cdot y} -1 ) a}_{H^2(Q)}^2 \\
& \le  C\norm{( \widetilde{a}_\ep - a )}_{H^2(Q)}^2 + \norm{(e^{ 2\pi{\rm i} ( \theta^\ep - \theta) \cdot y} -1 ) a}_{H^2(Q)}^2,
\end{flalign*}
and recalling that $e^{ 2\pi{\rm i} x}$ is uniformly continuous with respect to $x$ on $[0,1]$. Similarly, one can show $b_\ep \rightarrow b$ strongly in $H^2(Q)$ as $\ep \rightarrow 0$. Therefore $\phi^\ep \rightarrow \phi$ strongly in $[H^1(Q)]^2$ as $\ep \rightarrow 0$. \\

{\bf Case 2:} $\theta^\ep \rightarrow 0$. Let us assume without loss of generality $\theta^\ep \neq 0$. First let us consider the case $\theta = 0$, $\phi \in V(0)$, $\mv{\phi}=0$. By Lemma \ref{lem:spcom2}(ii) $\phi^0 = \nabla a + \nabla^{\perp} b$ for some $a,b \in H^2_{\#}(Q)$ such that 
\begin{alignat*}{3}
\Delta a  = f_1, & \hspace{3cm} & \Delta b  =f_2,
\end{alignat*}
for given $f_1,f_2$. By setting $\phi^\ep = \nabla a_\ep + \nabla^{\perp} b_\ep$, where $a_\ep,b_\ep \in H^2_{\theta^\ep}(Q)$ solve 
\begin{alignat*}{3}
\Delta a_\ep  = f_1, & \hspace{3cm} & \Delta b_\ep  =f_2,
\end{alignat*}
one can show by arguments very similar to those presented above that $\phi^\ep \rightarrow \phi$ strongly in $[H^1(Q)]^2$.

It remains to consider $\phi = c \in \mathbb{C}^2$. We shall show that there exists a sequence of $\theta^\ep$-quasi periodic functions $\phi^\ep$ that converges strongly to the constant vector $c$ in $[H^1(Q)]^2$. This requires the construction of special functions as follows: Denote by $N_\ep \in H^2_{\theta^\ep}(Q)$ the unique solution to
\begin{equation}
\label{eq:spcom8}
\Delta N_\ep = 2 \pi \vert \theta^\ep \vert \chi_0,
\end{equation} 
where $\chi_0$ is the characteristic function of $Q_0$. Now, let us introduce $\phi^\ep = c_1^\ep u^\ep_1 + c_2^\ep  u^\ep_2$, where $u^\ep_1 = \nabla N_\ep$, $u^\ep_2 = \nabla^{\perp} N_\ep$ and $c^\ep_1,c^\ep_2$ are constants yet to be determined. We shall show that $\phi^\ep \rightarrow \phi =c$ strongly in $[H^1(Q)]^2$ as $\ep \rightarrow 0$ for  specially  chosen constants $c^\ep_1,c^\ep_2$.\\
\noindent By the representation $N_\ep(y) = e^{ 2\pi{\rm i} \theta^\ep \cdot y} M_\ep(y)$, for some  $M_\ep \in H^2_{\#}(Q)$, we see by \eqref{eq:spcom8} that
\begin{equation}
\label{eq:spcom9}
\Delta M_\ep + 4 \pi i \theta^\ep \cdot \nabla M_\ep - 4 \pi^2 \vert \theta^\ep \vert^2 M_\ep = 2 \pi \vert \theta^\ep \vert \chi_0 e^{- 2\pi{\rm i} \theta^\ep \cdot y}.
\end{equation}
Since  $M_\ep = C_\ep + \widetilde{M}_\ep$, for some constant $C^\ep$ and some function $\widetilde{M}_\ep \in H^2_{\#}(Q)$, $\mv{\widetilde{M}^\ep}=0$, we find by \eqref{eq:spcom9} that
$$
C_\ep = - \frac{1}{2\pi \vert \theta^\ep \vert} \int_{Q_0}e^{- 2\pi{\rm i} \theta^\ep \cdot y}
$$
and
\begin{equation}
\label{eq:spcom10}
\Delta \widetilde{M}_\ep + 4 \pi{\rm i }\theta_\ep \cdot \nabla \widetilde{M}_\ep - 4 \pi^2 \vert \theta^\ep \vert^2 \widetilde{M}_\ep = f_\ep,
\end{equation}
where $f_\ep : = 2\pi \vert \theta^\ep \vert \chi_0 e^{- 2\pi{\rm i} \theta^\ep \cdot y} + 4\pi^2 \vert \theta^\ep \vert^2 C_\ep$. It is clear $\mv{f^\ep} = 0$, therefore, by Lemma \ref{lem:spcom3.1}(ii), $\norm{\widetilde{M}_\ep}_{H^2(Q)} \le C \norm{f_\ep}_{L^2(Q)}$ for some $C>0$ independent of $\ep$. Hence,  one can show that $\widetilde{M}_\ep$ strongly converges to zero in $H^2_{\#}(Q)$ as $\ep \rightarrow 0$.  Indeed, this follows from the fact that $f_\ep \rightarrow 0$ strongly in $L^2(Q)$ as $\ep \rightarrow 0$, which is seen by an application of the dominated convergence theorem upon noticing that $\vert f_\ep \vert \le 2 \pi$, $f_\ep \rightarrow 0$ a.e. in $Q$ as $\ep \rightarrow 0$. 

  Now, we find that
\begin{multline*}
\phi^\ep = \left( -{\rm i} c^\ep_1 \frac{\theta^\ep}{\vert \theta^\ep \vert}  -{\rm i} c^\ep_2 \frac{{\theta^\ep}^{\perp}}{\vert \theta^\ep \vert} \right)\left( \int_{Q_0}{e^{- 2 \pi{\rm i} \theta^\ep \cdot y}} \right) e^{ 2 \pi {\rm i} \theta^\ep \cdot y} \   \\  + c^\ep_1 \nabla \left( e^{ 2\pi {\rm i} \theta^\ep \cdot y} 
 \widetilde{M}_\ep \right) + c^\ep_2 \nabla^{\perp} \left( e^{ 2\pi {\rm i} \theta^\ep \cdot y}
 \widetilde{M}_\ep \right).
\end{multline*}
Note that $\left( \int_{Q_0}{e^{-2\pi{\rm i} \theta^\ep \cdot y}}\right) e^{ 2\pi \theta^\ep {\rm i}\cdot y} \rightarrow 1$ uniformly in $\overline{Q}$, and $ e^{ 2\pi{\rm i} \theta^\ep \cdot y}\widetilde{M}_\ep \rightarrow 0$  strongly in $H^2(Q)$ as $\ep \rightarrow 0$. Therefore, to show $\phi^\ep \rightarrow c$ strongly in $[H^1(Q)]^2$ it is sufficient to set
$$
 -{\rm i} c^\ep_1 \frac{\theta^\ep}{\vert \theta^\ep \vert} -{\rm i} c^\ep_2 \frac{{\theta^\ep}^{\perp}}{\vert \theta^\ep \vert} = c,
$$
i.e. to choose
\begin{alignat*}{3}
c^\ep_1 =  \frac{{\rm i}}{\vert \theta^\ep \vert} c \cdot \theta^\ep,& \hspace{3cm} & c^\ep_2  =  \frac{{\rm i}}{\vert \theta^\ep \vert} c \cdot {\theta^\ep}^{\perp}.
\end{alignat*}
Hence, $c^\ep_1, c^\ep_2$ are uniformly bounded in $\ep$ and 
\begin{equation*}
\phi^\ep =c  \left( \int_{Q_0}e^{- 2 \pi{\rm i} \theta^\ep \cdot y} \right) e^{ 2 \pi{\rm i} \theta^{\ep} \cdot y}  + c^\ep_1 \nabla \left( e^{ 2\pi {\rm i} \theta^\ep \cdot y} 
 \widetilde{M}_\ep \right) + c^\ep_2 \nabla^{\perp} \left( e^{ 2\pi {\rm i} \theta^\ep \cdot y}
 \widetilde{M}_\ep \right) \longrightarrow c,
\end{equation*}
strongly in $[H^1(Q)]^2$ as $\ep \rightarrow 0$. This completes the proof of the lemma.
\end{proof}

\begin{lem}[Continuity of $\A(\theta)$]
\label{apb.1}
For any $\theta \in [0,1)^2$, let $\theta^\ep\in[0,1)^2$ be such that $\theta^\ep \rightarrow \theta$ as $\ep \rightarrow 0$. Then the sequence 
$\big(\A(\theta^\ep) + I\big)^{-1}$ converges to $\big(\A(\theta) + I \big)^{-1}$ in the operator norm. In particular, the eigenvalues $\lambda_k(\theta),$ $k\in{\mathbb N}$, of $\A(\theta)$  are continuous functions of  $\theta\in[0,1)^2,$ {\it i.e.} $\lim_{\ep \rightarrow 0}\lambda_k(\theta^\ep) = \lambda_k(\theta).$
\end{lem}

In establishing the proof of Lemma \ref{apb.1} we shall use the following result.

\begin{lem}
\label{apb.lem1}
For any $\theta \in [0,1)^2$, let $\theta^\ep$ be such that $\theta^\ep \rightarrow \theta$  and let $f^\ep, f \in [L^2(Q)]^2$ be such that $f^\ep \rightharpoonup f$ in 
$[L^2(Q)]^2$ as $\ep \rightarrow 0$. Then, the sequence of solutions $u^\ep \in V(\theta^\ep)$ of the problems 
\begin{multline*}
\int_Q \nabla u_1^\ep \cdot \overline{\nabla \phi_1} + \gamma^{-1} \left( \div u^\ep \cdot \overline{\div \phi} + \div {u^\ep}^\perp \cdot \overline{\div \phi^\perp}\right) + \int_Q \epsilon_1 \left( u^\ep \cdot \phi + \gamma \chi_0 u^\ep_1 \phi_1 \right) \\ =  \int_Q \epsilon_1 \left( f^\ep \cdot \phi + \gamma \chi_0 f^\ep_1 \phi_1 \right), \qquad \forall \phi \in V(\theta^\ep),
\end{multline*}
weakly converges in $[H^1(Q)]^2$
to the solution $u \in V(\theta)$ of the problem
\begin{multline*}
\int_Q \nabla u_1 \cdot \overline{\nabla \phi_1} + \gamma^{-1} \left( \div u \cdot \overline{\div \phi} + \div  u ^\perp \cdot \overline{\div \phi^\perp}\right) + \int_Q \epsilon_1 \left( u \cdot \phi + \gamma \chi_0 u_1 \phi_1 \right) \\ =  \int_Q \epsilon_1 \left( f \cdot \phi + \gamma \chi_0 f_1 \phi_1 \right), \qquad \forall \phi \in V(\theta).
\end{multline*}
\end{lem}

\begin{proof}
Let $u^\ep \in V(\theta^\ep)$ be the solution to
\begin{multline}
\label{apb.ee1}
\int_Q \nabla u_1^\ep \cdot \overline{\nabla \phi_1} + \gamma^{-1} \left( \div u^\ep \cdot \overline{\div \phi} + \div {u^\ep}^\perp \cdot \overline{\div \phi^\perp}\right) + \int_Q \epsilon_1 \left( u^\ep \cdot \phi + \gamma \chi_0 u^\ep_1 \phi_1 \right) \\ =  \int_Q \epsilon_1 \left( f^\ep \cdot \phi + \gamma \chi_0 f^\ep_1 \phi_1 \right), \qquad \forall \phi \in V(\theta^\ep).
\end{multline}
As the sequence $f^\ep$ is weakly convergent, it is bounded and by choosing $\phi = u^\ep$ in \eqref{apb.ee1}, we find that  $\norm{u^\ep}_{[H^1(Q)]^2}\le C\norm{f^\ep}_{[L^2(Q)]^2}$ for some constant $C>0$,
{\it i.e.} the sequence $u^\ep$ is bounded in $[H^1(Q)]^2.$ In particular, up to a subsequence, $u^\ep$ converges weakly in $[H^1(Q)]^2$ (hence strongly in $[L^2(Q)]^2$) to some $u \in [H^1(Q)]^2$. Moreover, it can be readily shown that $u \in V(\theta)$. 

Furthermore, for a fixed $\phi \in V(\theta)$ there exists, by Lemma \ref{lem:spcom1}, $\phi^\ep \in V(\theta^\ep)$ such that $\phi^\ep \rightarrow \phi$ strongly in $[H^1(Q)]^2$ as $\ep \rightarrow 0$. Choosing $\phi^\ep$ as the test function in \eqref{apb.ee1} and passing to the limit $\ep \rightarrow 0$ shows that  $u$ is a solution to 
\begin{multline*}
\int_Q \nabla u_1 \cdot \overline{\nabla \phi_1} + \gamma^{-1} \left( \div u \cdot \overline{\div \phi} + \div  u ^\perp \cdot \overline{\div \phi^\perp}\right) + \int_Q \epsilon_1 \left( u \cdot \phi + \gamma \chi_0 u_1 \phi_1 \right) \\ =  \int_Q \epsilon_1 \left( f \cdot \phi + \gamma \chi_0 f_1 \phi_1 \right), \qquad \forall \phi \in V(\theta).
\end{multline*}
By virtue of the fact that the solution $u$ is unique, see \eqref{normofV}, the above result holds for any subsequence of $u^\ep$ and therefore holds for the whole sequence $u^\ep$.
\end{proof}
\begin{proof}[Proof of Lemma \ref{apb.1}.]
\noindent \textit{Step 1:} Let $\theta, \theta^\ep$ satisfy the assumptions of the lemma. First we show that the operator sequence 
$R^\ep:=\big( \A(\theta^\ep) + I \big)^{-1}$ converges uniformly to $R:= \big(\A(\theta) + I \big)^{-1}$ as $\ep \rightarrow 0$, {\it i.e.} 
$$
\norm{R^\ep - R}{} = \sup_{\norm{f}_{[L^2(Q)]^2}=1}  \norm{R^\ep f - R f }_{[L^2(Q)]^2} \rightarrow 0 \quad \text{as} \ \ep \rightarrow 0.
$$
For all $\ep$, let $f^\ep$ be such that
$$
\sup_{\norm{f}_{[L^2(Q)]^2}=1}  \norm{R^\ep f - R f }_{[L^2(Q)]^2}\le\norm{R^\ep f^\ep - R f^\ep }_{[L^2(Q)]^2}+ \ep.
$$
Since $\norm{f^\ep}_{[L^2(Q)]^2} = 1$, the sequence $f^\ep$ has a subsequence that converges weakly to some $f\in [L^2(Q)]^2.$ Therefore, by Lemma \ref{apb.lem1} the 
sequence $R^\ep f^\ep$ converges to $R f$ strongly in $[L^2(Q)]^2.$ Furthermore, since $R$ is compact, we infer that $R f^\ep$ converges to $R f$ strongly in $[L^2(Q)]^2,$ and therefore
$$
\sup_{\norm{f}_{L^2(Q)}=1}\norm{R^\ep f - R f }_{[L^2(Q)]^2}\le\norm{R^\ep f^\ep - R f }_{[L^2(Q)]^2} + \norm{R f^\ep - R f }_{[L^2(Q)]^2} + \ep.
$$
The right-hand side of the above estimate converges to zero as $\ep \rightarrow 0.$ The result follows as this argument holds for any subsequence of $f^\ep$ and therefore for the whole sequence $f$. \\

\noindent \textit{Step 2:} We shall now show the continuity in $\theta$ of the eigenvalues $\lambda_k(\theta).$ To this end we establish that the eigenvalues 
$\mu_k(\theta)=\bigl(\lambda_k(\theta) + 1\bigr)^{-1}$ of the operator $\bigl(\A(\theta)+I\bigr)^{-1}$ are continuous in $\theta.$ 
To prove that $\mu_k(\theta)$ are continuous we note that for any $f \in [L^2(Q)]^2$ one has
$$
\frac{(R^\ep f ,f )_{[L^2(Q)]^2}}{\norm{f}_{[L^2(Q)]^2}^2} - \norm{R^\ep - R}{} \le \frac{(R f ,f )_{[L^2(Q)]^2}}{\norm{f}_{[L^2(Q)]^2}^2} \le \frac{(R^\ep f ,f )_{[L^2(Q)]^2}}{\norm{f}_{[L^2(Q)]^2}^2} + \norm{R^\ep - R}{},
$$
and  by the min-max variational principle 
$$
\mu_k(\theta) = \inf_{\substack{ F \subset [L^2(Q)]^2, \\ \text{dim}F =k}}\ \  \sup_{\substack{ f \in F, \\ \norm{f}_{[L^2(Q)]^2}=1}} \frac{(Rf ,f )_{[L^2(Q)]^2}}{\norm{f}_{[L^2(Q)]^2}^2}
$$
implies that $\bigl\vert \mu_k(\theta^\ep)-\mu_k(\theta)\bigr\vert \le\norm{R^\ep - R}{}.$ Invoking the uniform convergence of $R^\ep$ to $R$ as $\ep \rightarrow 0$ proves the result.
\end{proof}


\section{Lower semicontinuity of the spectrum under homogenisation}
In this section we shall be concerned with proving the spectral inclusion
\begin{equation}
\label{semicontspect}
\lim_{\ep \rightarrow 0} \sigma_\ep \supset \bigcup_{\theta \in [0, 1)^2} \sigma(\theta) = \sigma_0,
\end{equation}
where,  $\sigma(\theta)$ is the spectrum of the self-adjoint operator $\A(\theta)$ introduced in Section \ref{secblochlim}.  More precisely, we shall  show that if $\lambda \in \sigma_0$, then there exists $\lambda_\ep \in \sigma_\ep$ such that $\lambda_\ep \rightarrow \lambda$ as $\ep \rightarrow 0$.

Before we proceed, we shall present a slightly different characterisation of the set $\sigma_\ep$. Formally, $\lambda \in \sigma_\ep$ if there exists non-zero $v$ such that
\begin{equation}
\label{effectivecontrast}
\int_{\mathbb{R}^2} A^\ep(y) \nabla v \cdot \overline{\nabla \phi} \ \mathrm{d}y = \lambda \int_{\mathbb{R}^2} \rho(y)v \cdot \overline{\phi}\ \mathrm{d}y  \quad \forall \phi \in [C^{\infty}_{0}(\mathbb{R}^2)]^2,
\end{equation}
where $A^\ep = \tfrac{1}{\ep^2}\chi_1(y) A^\ep_{1} + \tfrac{1}{\epsilon_0 - \epsilon_1 + \ep^2} \chi_0(y)A^\ep_{0}$ and the second order constant symmetric tensors $A^\ep_{r},$ $r=0,1$ are given in matrix representation form $\{ A^\ep_{r} \}_{ijpq} = \{ A^{jq}_{\ep,r} \}_{ip}$,
\begin{align}
\label{tensor2}
A^{11}_{\ep,r} = \left(
\begin{matrix}
\epsilon_r & 0 \\
0 & \mu 
\end{matrix}
\right),  &  & 
A^{22}_{\ep,r}  = \left(
\begin{matrix}
\epsilon_r & 0 \\
0 & \mu 
\end{matrix}
\right), & &  A^{12}_{\ep,r} = - A^{21}_{\ep,r}  =  \left(
\begin{matrix}
0 &  \sqrt{\mu(\epsilon_1 - \ep^2)} \\
- \sqrt{\mu(\epsilon_1 - \ep^2)} & 0 
\end{matrix}
\right). 
\end{align} 
Henceforth, for a second order symmetric tensor $A$ that is given as follows
\begin{equation}
\label{tensor1}
A = \left( \begin{matrix}
A^{11} & A^{12} \\ A^{21} & A^{22}
\end{matrix} \right),
\end{equation}
it is to be understood that $A_{ijpq} = A^{jq}_{ip}$. Moreover, for a given matrix $B$ we introduce the vectors $b^{j}$, $j=1,2$, given by $b^j_i = B_{ij}$, then the tensor contraction $C = A B$  becomes a matrix vector product as follows
$$
\left( \begin{matrix}
c^1 \\ c^2 
\end{matrix} \right) = \left( \begin{matrix}
A^{11} & A^{12} \\ A^{21} & A^{22}
\end{matrix} \right) \left( \begin{matrix}
b^1 \\ b^2
\end{matrix} \right)= \left( \begin{matrix}
A^{11}b^1 + A^{12}b^2 \\ A^{21}b^1 + A^{22}b^2
\end{matrix} \right).
$$

For the non-trivial solution $v$ to \eqref{effectivecontrast}, and test function $\phi$, we perform the following scaling $u = (\sqrt{\epsilon_1} v_1 , \sqrt{\mu} v_2)$, $\varphi = (\sqrt{\epsilon_1} \phi_1 , \sqrt{\mu} \phi_2)$. Then, by a change of variables $x = \ep y$, \eqref{effectivecontrast} becomes: Find $u$ such that
\begin{equation}
\label{effectivecontrast2}
\int_{\mathbb{R}^2} A^\ep\left(\tfrac{x}{\ep}\right) \nabla u \cdot \overline{\nabla \phi} \ \mathrm{d}x =  \lambda \int_{\mathbb{R}^2} \rho\left(\tfrac{x}{\ep}\right)u \cdot \overline{\phi} \ \mathrm{d}x  \quad \forall \phi \in [C^{\infty}_{0}(\mathbb{R}^2)]^2.
\end{equation}
where 
\begin{equation}
\label{rho}
\rho(y) = \chi_1(y)\left( \begin{array}{cc} 1 & 0 \\ 0 & 1\end{array} \right) + \chi_0(y) \left( \begin{array}{cc} \frac{\epsilon_0}{\epsilon_1} & 0 \\ 0 & 1\end{array} \right), 
\end{equation}
$A^\ep(y) = \chi_1(y) A^\ep_{1} + \tfrac{\ep^2}{\epsilon_0 - \epsilon_1 + \ep^2} \chi_0(y)A^\ep_{0}$, and
\begin{align}
A^{\ep}_1 & = \left( \begin{matrix} 1 & 0 & 0 & \sqrt{(1 - \frac{\ep^2}{\epsilon_1})} \\ 0 & 1 & - \sqrt{(1 - \frac{\ep^2}{\epsilon_1})} & 0 \\
0 & - \sqrt{(1 - \frac{\ep^2}{\epsilon_1})} & 1 & 0 \\ \sqrt{(1 - \frac{\ep^2}{\epsilon_1})} & 0 & 0 &1 \end{matrix} \right),  \label{emextraeq1}\\
A^{\ep}_0 & = \left( \begin{matrix} \frac{\epsilon_0}{\epsilon_1} & 0 & 0 & \sqrt{(1 - \frac{\ep^2}{\epsilon_1})} \\ 0 & 1 & - \sqrt{(1 - \frac{\ep^2}{\epsilon_1})} & 0 \\
0 & - \sqrt{(1 - \frac{\ep^2}{\epsilon_1})} & \frac{\epsilon_0}{\epsilon_1} & 0 \\ \sqrt{(1 - \frac{\ep^2}{\epsilon_1})} & 0 & 0 &1 \end{matrix} \right). \label{emextraeq2}
\end{align} 
Thus, we conclude that the set $\sigma_\ep$ coincides with the spectrum of the high-contrast operator $\A^\ep : [H^1(\RR^2)]^2 \rightarrow [L^2(\RR^2)]^2$, 
generated by the bilinear form
$$
\a_\ep(u,v): = \int_{\mathbb{R}^2} A^\ep\left(\tfrac{x}{\ep}\right) \nabla u \cdot \overline{\nabla v} \ \mathrm{d}x,
$$
such that the mapping $\A^\ep u = f$ relates $f$ to the unique solution $u$ of
$$
\a_\ep(u,\phi) = \int_{\RR^2} \rho\left(\tfrac{x}{\ep}\right) f \cdot \overline{\phi}, \quad \forall \phi \in [C^\infty_0(\RR^2)]^2.
$$
The remainder of this section is devoted to studying the strong  two-scale resolvent limit of the operator $\A^\ep$.
\subsection{Periodic homogenisation}\label{sec:PChom}
Here we analyse the strong  $Q$-periodic  two-scale resolvent limit of the operator $\A^\ep$. That is, for a given $f^\ep \in [L^2(\RR^2)]^2$ two-scale converging to $f \in [L^2(\RR^2 \times Q)]^2$, we wish to study the two-scale limit of the sequence of solutions $u^\ep \in \left[H^1(\mathbb{R}^2)\right]^2$ to 
\begin{multline}
\label{finalproblem}
\int_{\mathbb{R}^2} \Big( \chi_1\left(\tfrac{x}{\ep}\right) A^\ep_{1} + \tfrac{\ep^2}{\epsilon_0 - \epsilon_1 + \ep^2} \chi_0\left(\tfrac{x}{\ep}\right)A^\ep_{0} \Big) \nabla u^\ep \cdot \overline{\nabla \phi} \ \mathrm{d}x +  \int_{\mathbb{R}^2} \rho\left(\tfrac{x}{\ep}\right) u^\ep \cdot \overline{\phi} \ \mathrm{d}x \\ = \int_{\mathbb{R}^2} \rho\left(\tfrac{x}{\ep}\right) f^{\ep} \cdot \overline{\phi} \ \mathrm{d}x  \quad \forall \phi \in \left[C^{\infty}_{0}(\mathbb{R}^2)\right]^2.
\end{multline}
Here $\rho$, $A^\ep_0$ and $A^\ep_1$ are given by \eqref{rho}-\eqref{emextraeq2}. It can readily seen that  $\big(\chi_1(y)A^\ep_1 + \chi_0(y)A^\ep_0\big)$ is (major) symmetric and uniformly positive definite in $y$ for sufficiently small  $\ep$. That is 
\begin{equation}
\label{emtensorprop}
\begin{aligned}
\big((\chi_1(y)A^\ep_1 + \chi_0(y)A^\ep_0)\big)_{ijpq} & =\big((\chi_1(y)A^\ep_1 + \chi_0(y)A^\ep_0)\big)_{pqij} \\
\big((\chi_1(y)A^\ep_1 + \chi_0(y)A^\ep_0)\big) \eta \cdot \eta & > \nu \vert \eta \vert^2, \quad \forall \eta \in \mathbb{C}^2, \forall y \in Q,
\end{aligned}
\end{equation} 
for some $\nu >0$ independent of $\ep$. In particular, \eqref{emtensorprop}  implies the existence and uniqueness of solutions to \eqref{finalproblem}. 

Introducing the following tensors
\begin{align}
\label{emdegentensor1}
a^{(1)}(y) & = \chi_1(y) \left( \begin{matrix}
1 & 0 & 0 & 1 \\ 0 & 1 & -1 & 0 \\ 0 & -1 & 1 & 0 \\ 1 & 0 & 0 & 1
\end{matrix}
\right), \\ a^{(0)}(y) & =\frac{1}{2 \epsilon_1}\chi_1(y)\left( \begin{matrix}
0 & 0 & 0 & - 1 \\ 0 & 0 & 1 & 0 \\ 0 & 1 & 0 & 0 \\ -1 & 0 & 0 & 0
\end{matrix} \right) + \frac{1}{(\epsilon_0 - \epsilon_1)}\chi_0(y) \left( \begin{matrix}
\frac{\epsilon_0}{\epsilon_1} & 0 & 0 & 1 \\ 0 & 1 & -1 & 0 \\ 0 & -1 & \frac{\epsilon_0}{\epsilon_1} & 0 \\ 1 & 0 & 0 & 1
\end{matrix} \right), \label{emdegentensor2}
\end{align}
we present our first important result
\begin{lem}
\label{emextralem1}
Let $u^\ep$ be the solution to \eqref{finalproblem}. Then there exist constants $C>0$, $ \ep_0 > 0$ such that for all $\ep < \ep_0$
\begin{align}
\norm{u^\ep}_{[L^2(\mathbb{R}^2)]^2} & \le C\norm{f^\ep}_{[L^2(\mathbb{R}^2)]^2} , \label{emextralemeq1}\\
\norm{\ep \nabla u^\ep}_{[L^2(\mathbb{R}^2)]^{2\times2}} & \le C\norm{f^\ep}_{[L^2(\mathbb{R}^2)]^2} , \label{emextralemeq2}\\
\norm{\left( a^{(1)}\left(\tfrac{x}{\ep}\right)\right)^{1/2} \nabla u^\ep}_{[L^2(\mathbb{R}^2)]^{2\times2}} & \le C\norm{f^\ep}_{[L^2(\mathbb{R}^2)]^2}   \label{emextralemeq3}.
\end{align}
Furthermore,
\begin{multline}
\lim_{\ep \rightarrow 0}\sup_{\substack{\phi \in [H^1(\RR^2)]^2 \\ \norm{\phi}_{H^1}=1}} \left\vert \int_{\mathbb{R}^2}  \left( \chi_1\left(\tfrac{x}{\ep}\right) A^\ep_{1} + \tfrac{\ep^2}{\epsilon_0 - \epsilon_1 + \ep^2} \chi_0\left(\tfrac{x}{\ep}\right)A^\ep_{0} \right) \nabla u^\ep \cdot \overline{\nabla \phi} \ \mathrm{d}x \right. \\ - \left. \int_{\mathbb{R}^2} \left[ a^{(1)}\left(\tfrac{x}{\ep}\right) + \ep^2 a^{(0)}\left(\tfrac{x}{\ep}\right) \right] \nabla u^\ep \cdot \overline{\nabla \phi} \ \mathrm{d}x \right\vert = 0. \label{emextralemeq4}
\end{multline}
\end{lem}

\begin{proof}
 Henceforth, while constants may change between successive inequalities we shall always denote them by the letter $C$. Taking $\phi = u^\ep$ in \eqref{finalproblem} gives
\begin{multline}
\label{emextraeq3}
\int_{\mathbb{R}^2}  \big( \chi_1\left(\tfrac{x}{\ep}\right) A^\ep_{1} + \tfrac{\ep^2}{\epsilon_0 - \epsilon_1 + \ep^2} \chi_0\left(\tfrac{x}{\ep}\right)A^\ep_{0} \big) \nabla u^\ep \cdot \overline{\nabla u^\ep} \ \mathrm{d}x  + \int_{\mathbb{R}^2} \rho\left(\tfrac{x}{\ep}\right) u^\ep \cdot \overline{u^\ep} \ \mathrm{d}x \\ = \int_{\mathbb{R}^2} \rho\left(\tfrac{x}{\ep}\right) f^{\ep} \cdot \overline{u^\ep} \ \mathrm{d}x.
\end{multline}
Using  \eqref{rho},\eqref{emextraeq3} and the Cauchy-Schwarz inequality we find that
\begin{flalign*}
 \int_{\mathbb{R}^2} \vert u^\ep \vert^2  & \le C \int_{\mathbb{R}^2} \rho\left(\tfrac{x}{\ep}\right) u^\ep \cdot \overline{u^\ep} \le C \int_{\mathbb{R}^2} \rho\left(\tfrac{x}{\ep}\right) f^\ep \cdot \overline{u^\ep} \\
& \le C \int_{\mathbb{R}^2}  f^\ep \cdot \overline{u^\ep} \le C \left( \int_{\mathbb{R}^2} \vert f^\ep \vert^2 \right)^{1/2}\left( \int_{\mathbb{R}^2} \vert u^\ep \vert^2 \right)^{1/2},
\end{flalign*}
resulting in \eqref{emextralemeq1}. Similarly \eqref{emextralemeq2} holds due to  \eqref{emtensorprop}, \eqref{emextralemeq1}, \eqref{emextraeq3}  and the observation that for sufficiently small $\ep$
\begin{flalign*}
\frac{\ep^2}{\epsilon_0 - \epsilon_1 + \ep^2} \int_{\mathbb{R}^2} \vert \nabla u^\ep \vert^2 & \le \frac{\nu \ep^2}{\epsilon_0 - \epsilon_1 + \ep^2} \int_{\mathbb{R}^2} \big(\chi_1\left(\tfrac{x}{\ep}\right)A^\ep_1 + \chi_0\left(\tfrac{x}{\ep}\right)A^\ep_0\big) \nabla u^\ep \cdot \overline{\nabla u^\ep} \\ & \le \nu \int_{\mathbb{R}^2} \left(\chi_1\left(\tfrac{x}{\ep}\right)A^\ep_1 + \frac{\ep^2}{\epsilon_0 - \epsilon_1 + \ep^2}\chi_0\left(\tfrac{x}{\ep}\right) A^\ep_0 \right) \nabla u^\ep \cdot \overline{\nabla u^\ep}.
\end{flalign*}

For fixed sufficiently small $\ep >0 $ such that $\ep < \epsilon_1$ and $\ep < (\epsilon_0 - \epsilon_1)$ we know that 
\begin{align*}
\left(1 - \frac{\ep^2}{\epsilon_1} \right)^{1/2} & = 1 - \frac{\ep^2}{2\epsilon_1} + O(\ep^4), \\
\frac{\ep^2}{\epsilon_0 - \epsilon_1 + \ep^2} & = \frac{\ep^2}{\epsilon_0 - \epsilon_1} + O(\ep^4).
\end{align*}
We further observe that, by \eqref{emextraeq1}, one has
\begin{multline*}
\int_{\mathbb{R}^2} \vert \left( a^{(1)}\left(\tfrac{x}{\ep}\right)\right)^{1/2} \nabla u^\ep \vert^2 = \int_{\mathbb{R}^2} a^{(1)}\left(\tfrac{x}{\ep}\right) \nabla u^ \ep \cdot \overline{\nabla u^\ep} =  \int_{\mathbb{R}^2} \chi_1\left(\tfrac{x}{\ep}\right) A^{\ep}_1 \nabla u^ \ep \cdot \overline{\nabla u^\ep}  \\ + 2\left(1 - \sqrt{1 - \tfrac{\ep^2}{\epsilon_1}} \right)\int_{\mathbb{R}^2} \chi_1\left(\tfrac{x}{\ep}\right) \left(  u^\ep_{1,1}u^\ep_{2,2} - u^\ep_{2,1}u^\ep_{1,2} \right).
\end{multline*}
Therefore \eqref{emextralemeq3} holds by \eqref{emextralemeq1}, \eqref{emextralemeq2}, \eqref{emextraeq3}  and the Cauchy-Schwarz inequality.

It remains to prove \eqref{emextralemeq4}.  For $A^\ep_0, A^\ep_1$, given by \eqref{emextraeq1}-\eqref{emextraeq2}, we find by direct calculation that
$$
\chi_1(y)A^\ep_1 + \frac{\ep^2}{\epsilon_0 - \epsilon_1 + \ep^2}\chi_0(y) A^\ep_0 = a^{(1)}(y) + \ep^2 a^{(0)}(y) + R^{\ep}(y),
$$
where
$$
R^{\ep}(y) = O(\ep^4).
$$
Equation \eqref{emextralemeq4} then follows from \eqref{emextralemeq2}.
\end{proof}

Recalling the space $V(\theta)$ introduced in Section \ref{secblochlim}, we denote by $V$ the space $V(0)$, i.e.
\begin{equation}
\label{emspaceV}
V = \left\{ v \in [H^1_{\#}(Q)]^2 : \text{ $\div{v} = 0$ and  $\div{v^{\perp}} = 0$ in  $Q_1$}\right\}.
\end{equation}

Following the general approach to partial degenerate PDEs of \cite{VPSIVK}, we introduce the following space
	\begin{gather}
	W\,:=\,\left\{\, \psi\,\in\,\left[L^2(Q)\right]^{2\times
		2}\,\left\vert \, \mbox{
		div}\left(\,\left(a^{(1)}\right)^{1/2}\psi\,\right)\,=\,0
	\ \mbox{ in  }
	\left[H_\#^{-1}(Q)\right]^2\,\right.\right\}, \label{eq:pdhom10}
	\end{gather}
and present the following result
\begin{lem}\label{lem:pdhom1}
Suppose $f^\ep \in [L^2(\RR^2)]^2$ two-scale converges to $f \in [L^2(\RR^2 \times Q)]^2$ and consider the sequence of solutions $u^\ep$ to problems \eqref{finalproblem}. Then, there exists $u(x,y)\in L^2\left(\RR^2;\, V\right)$ and
	$\xi(x,y)\in L^2\left(\RR^2; \, W\right)$ such that, up
	to extracting a subsequence in $\varepsilon$ which we do not
	relabel,
	\begin{eqnarray}
	u^\varepsilon&\stackrel{2}\rightharpoonup& u(x,y)         \label{eq:pdhom6} \\
	\varepsilon \nabla u^\varepsilon  &\stackrel{2}\rightharpoonup& \nabla_y u(x,y)           \label{eq:pdhom7} \\
	\left( a^{(1)}\left(\tfrac{x}{\ep}\right)\right)^{1/2} \nabla u^\varepsilon\,   &\stackrel{2}\rightharpoonup& \xi(x,y).
	\label{eq:pdhom8}
	\end{eqnarray}
Henceforth, we use the symbol $\twoscale$ to denote two-scale convergence with respect to the periodic reference cell $Q$.
Moreover, the following identity holds: 
\begin{multline}
\label{eq:pdhom13} 
 \dblint{\RR^2}{\xi(x,y)\,\cdot\,\Psi(x,y)}= 
-\dblint{\RR^2}{u(x,y) \cdot \left(a^{(1)}(y)\right)^{1/2} {\rm div}_x{
		\Psi(x,y)}}, \\ \forall \Psi(x,y)\in C^\infty(\RR^2;\,W \cap [C^\infty_{\#}(Q)]^{2 \times 2}) .
\end{multline}
\end{lem}
\begin{proof}
	According to the theorem on relative (weak) two-scale
	compactness of a bounded sequences in $L^2$, see e.g.
	\cite{nguetseng,allaire,zhikov1}, the  estimates
	\eqref{emextralemeq1}-\eqref{emextralemeq3} imply, up to a strict\footnote{We shall see in Subsection \ref{sec:quasiotwoscale} below that the entire sequence does not converge. In particular, $u^\ep$ does not strongly two-scale converge to $u$. } subsequence in
	$\varepsilon$ which we shall no relabel, their  exists  $u \in \left[L^2\left(\RR^2; H^1_{\#}(Q)\right)\right]^{2}$, $\eta, \xi\in \left[L^2\left(\RR^2 \times
	Q\right)\right]^{2\times 2}$  such that  
\begin{eqnarray*}
u^\varepsilon&\stackrel{2}\rightharpoonup& u(x,y),   \\
\varepsilon \nabla u^\varepsilon  &\stackrel{2}\rightharpoonup& \eta(x,y),        \\
\left( a^{(1)}\left(\tfrac{x}{\ep}\right)\right)^{1/2} \nabla
u^\varepsilon\,   &\stackrel{2}\rightharpoonup& \xi(x,y).
\end{eqnarray*}
	
	Furthermore, by the definition of two-scale convergence, for fixed $\phi \in C^\infty_0 (\RR^2)$, $\varphi \in [C^\infty_{\#} (Q)]^2$ and $\psi \in [C^\infty_{\#} (Q)]^{2 \times 2}$ one have
	\begin{gather*}
	\lim_{\ep \rightarrow 0}\int_{\RR^2} u^\ep(x) \cdot \phi(x) \varphi(x / \ep) \ \mathrm{d}x = \dblint{\RR^2}{u(x,y) \cdot \phi(x) \varphi(y)} \\ 
	\lim_{\ep \rightarrow 0}\int_{\RR^2} \ep \nabla u^\ep(x) \cdot \phi(x) \psi(x / \ep) \ \mathrm{d}x =  \dblint{\RR^2}{\eta(x,y) \cdot \phi(x) \psi(y)}.
	\end{gather*}
	Noticing, via integration by parts, that
	\begin{multline*}
	\int_{\RR^2} \ep \nabla u^\ep(x) \cdot \phi(x) \psi(x / \ep) \ \mathrm{d}x = - \ep \int_{\RR^2}  u^\ep(x) \cdot \psi(x / \ep)   \nabla_x\hspace{0pt} \phi(x)\ \mathrm{d}x  \\ -
	\int_{\RR^2}  u^\ep(x) \cdot \phi(x) {\rm div}_y\hspace{0pt}{\psi(x / \ep)} \ \mathrm{d}x \overset{\ep \rightarrow 0}{\longrightarrow }  - \dblint{\RR^2}{ u(x,y) \cdot  \phi(x) {\rm div}_y{\psi(y)}}, 
	\end{multline*}
	we find
	\begin{equation}
	\label{addedyreg}
	\dblint{\RR^2}{\eta(x,y) \cdot \phi(x) \psi(y)} =  - \dblint{\RR^2}{ u(x,y) \cdot \phi(x){\rm div}_y{\psi(y)}}, 
	\end{equation}
	where for a matrix $M$, we define its divergence by $\{ {\rm div}M \}_i : = M_{i j, j}$. Therefore, the arbitrariness of $\phi, \psi$ and equation \eqref{addedyreg} imply that $u(x,y) \in \left[L^2(\RR^2 ; H^1_{\#}(Q))\right]^2$ with weak derivative $\nabla_y u(x,y) = \eta(x,y).$
	
	We argue that $u\in L^2(\RR^2 ; V)$, see \eqref{emspaceV}. Indeed, by \eqref{emextralemeq3} one has for a fixed $\Phi \in [C^\infty_0(\RR^2 ; C^\infty_{\#}(Q))]^{2\times2}$ that
	$$
	\lim_{\ep \rightarrow 0}\int_{\RR^2} \ep a^{(1)} \left(\tfrac{x}{\ep}\right) \nabla u^\ep \cdot  \Phi\left(x , \tfrac{x}{\ep} \right) = 0.
	$$
	Yet, on the other hand, by \eqref{eq:pdhom7} one has
	$$
	\lim_{\ep \rightarrow 0}\int_{\RR^2} \ep a^{(1)} \left(\tfrac{x}{\ep}\right) \nabla u^\ep \cdot \Phi\left(x , \tfrac{x}{\ep} \right) \ \mathrm{d}x = \int_{\RR^2}\int_Q a^{(1)}(y) \nabla_y u(x,y) \cdot \Phi\left(x , y \right) \ \mathrm{d}y \mathrm{d}x.
	$$
	By the arbitrariness of $\Phi$, we deduce that $a^{(1)} \nabla_y u = 0$ a.e. $x \in \RR^2$ and therefore, by \eqref{emdegentensor1}, $u( x , \cdot) \in V$.

	Similarly, one can show that  $\xi(x,y)\in L^2\left(\RR^2; \,
	W\right)$. Indeed, by choosing, in \eqref{finalproblem}, test functions of the form
	$\phi^\varepsilon(x)=\,
	\varepsilon\,\phi\left(x,\frac{x}{\varepsilon}\right)$ for
	any $\phi(x,y)\in \left[C^\infty_0\left(\RR^2;
	C_\#^\infty(Q)\right)\right]^{2}$, we notice via \eqref{emextralemeq2}-\eqref{emextralemeq4}, that 
	\begin{equation*}
	\lim_{\varepsilon \to 0} \int_{\RR^2} \Big(
	\chi_1 \left(\frac{x}{\varepsilon}\right) A^\ep_1   + \frac{\e^2}{\epsilon_0 - \epsilon_1 + \ep^2} \chi_0 \left(\frac{x}{\varepsilon}\right)A^\ep_0 \Big) \nabla
	u^\varepsilon(x) \,\cdot\,\varepsilon\nabla \phi\left(x,
	\frac{x}{\varepsilon}\right) \,dx\, = 0
	\end{equation*}
	In particular, by  \eqref{emextralemeq3} and \eqref{emextralemeq4}, one has
	\begin{multline*}
	\lim_{\varepsilon \to 0} \int_{\RR^2} 
	a^{(1)}\left(\frac{x}{\varepsilon}\right)  \nabla
	u^\varepsilon(x)  \,\cdot\,\nabla_y \phi\left(x,
	\frac{x}{\varepsilon}\right) \,dx\, \\
	= \lim_{\varepsilon \to 0} \int_{\RR^2} 
	a^{(1)}\left(\frac{x}{\varepsilon}\right)  \nabla
	u^\varepsilon(x)  \,\cdot\, \left[ \varepsilon\nabla \phi\left(x,
	\frac{x}{\varepsilon}\right) - \varepsilon\nabla_x \phi\left(x,
	\frac{x}{\varepsilon}\right) \right]\,dx  = 0
	\end{multline*}
	Yet, by \eqref{eq:pdhom8}, one finds that
		\begin{equation*}
		\lim_{\varepsilon \to 0} \int_{\RR^2} 
		a^{(1)}\left(\frac{x}{\varepsilon}\right)  \nabla
		u^\varepsilon(x)  \,\cdot\,\nabla_y \phi\left(x,
		\frac{x}{\varepsilon}\right) \,dx\, =\int_{\RR^2}\int_Q \xi \left(x , y \right)  \cdot \nabla_y 
		\phi\left(x , y \right) \ \mathrm{d}y \mathrm{d}x
		\end{equation*}
	Hence, by the arbitrariness of $\phi(x,y)$ one has
	$$
   \int_Q \xi \left(x , y \right)  \cdot \nabla_y \varphi\left( y \right) \ \mathrm{d}y = 0, \qquad \forall \varphi \in [C^\infty_{\#}(Q)]^{2} \text{	for a.e. $x \in \RR^2$},
	$$
	i.e  $\xi(x,y)\in L^2\left(\RR^2;\,W\right)$, see
	\eqref{eq:pdhom10}.

	It remains to show that $u$ and $\xi$ are related by \eqref{eq:pdhom13}. As smooth functions with compact support are  easily shown to approximate smooth functions with respect to the $L^2(\RR^2)$ norm, we shall prove \eqref{eq:pdhom13} for fixed $\Psi(x,y)\in C^\infty_0(\RR^2;\,W \cap [C^\infty_{\#}(Q)]^{2 \times 2})$. For such $\Psi$, integration by parts and the chain rule\footnote{ The justification of the multi-variable chain rule formula for functions of the form $f(x , x / \ep)$ where $f$ belongs $L^2(\RR^2 \times Q)$ can readily be justified by standard approximation arguments.} yields
	\begin{flalign*}
 \lim_{\ep \rightarrow 0}	\int_{\RR^2} \left( a^{(1)}\left(\tfrac{x}{\ep}\right)\right)^{1/2}  \nabla u ^ \ep \cdot \Psi \left( x ,  \tfrac{x}{\ep} \right) =   \lim_{\ep \rightarrow 0} 	\int_{\RR^2} - u ^ \ep \cdot  \left( a^{(1)}\left(\tfrac{x}{\ep}\right)\right)^{1/2} {\rm div}_x \Psi \left( x ,  \tfrac{x}{\ep} \right) ,
	\end{flalign*}
	which, via \eqref{eq:pdhom6} and \eqref{eq:pdhom8} yields \eqref{eq:pdhom13}.
\end{proof}


Continuing the general scheme of \cite{VPSIVK}, let us now introduce the space
\begin{multline}
U  := \Bigg\{ u(x,y) \in L^2(\mathbb{R}^2 ;V) : \exists \xi(x,y) \in L^2(\mathbb{R}^2 ; W) \text{ such that }, \biggr. \\
\dblint{\RR^2}{\xi(x,y)\,\cdot\,\Psi(x,y)}= 
-\dblint{\RR^2}{u(x,y) \cdot \left(a^{(1)}(y)\right)^{1/2} {\rm div}_x{
		\Psi(x,y)}}, \\ \forall \Psi(x,y)\in C^\infty(\RR^2;\,W \cap [C^\infty_{\#}(Q)]^{2 \times 2})  \Bigg\}, \label{emspaceU} 
\end{multline}
and the operator $T:U \rightarrow L^2$ defined by $Tu = \xi$ for $\xi$.
\begin{rem}
\label{Weylsdecomp}
Note that if $u \in U$ belongs to $C^\infty_0(\RR^2 ; V)$ then by \eqref{emspaceU} and the Generalised Weyl's decomposition, see Appendix, one has
$$
T u(x,y) = (a^{(1)}(y))^{1/2} \left( \nabla_x u(x,y) + \nabla_y u_1(x,y) \right),
$$
for some $u_1 \in [C^\infty(\RR^2 ; L^2(Q))]^2$.
\end{rem}
In fact, we find that for $a^{(1)}$   given by \eqref{emdegentensor1}, the map $T\ \equiv 0$; i.e.  the following result holds
\begin{prop}
\label{prop:emgenflux}
Let $T:U \rightarrow L^2$ be defined by $Tu = \xi$ for $\xi$ given in \eqref{emspaceU}. Then
$$
Tu = 0, \quad \forall u \in U.
$$
\end{prop}

\begin{proof}
To prove the proposition, it is sufficient to show that for all $u,v \in U$
$$
\dblint{\mathbb{R}^2}{Tu(x,y) \cdot T v(x,y)} = 0.
$$
Let us state here, and demonstrate below, the following property of functions belonging to $W$, see \eqref{eq:pdhom10}:
\begin{align}
\label{eq:emgenfluxprop1}
\text{$\psi_{11} + \psi_{22}= 0$ and $\psi_{12} - \psi_{21}= 0$ in $Q_1$}, \quad \forall \psi  \in W.
\end{align}
For fixed $u\in U$, $Tu \in L^2(\mathbb{R}^2 ; W)$ and therefore \eqref{emdegentensor1} and \eqref{eq:emgenfluxprop1} imply $a^{(1)}Tu = 0$. Hence,  by additionally noting that $a^{(1)}= \tfrac{1}{\sqrt{2}}(a^{(1)})^{1/2}$, one has $(a^{(1)})^{1/2}Tu = 0$. Now, by the Generalised Weyl's decomposition, see Remark \eqref{Weylsdecomp},  for $ \phi \in C^{\infty}_0(\mathbb{R}^2;V)$, $T\phi(x,y) = (a^{(1)}(y))^{1/2}(\nabla_x \phi(x,y) + \nabla_y \phi_1(x,y))$ for some $\phi_1 \in C^\infty\left( \mathbb{R}^2 ; W^{\perp}\right)$. Therefore, one has
$$
\dblint{\mathbb{R}^2}{Tu(x,y) \cdot T \phi(x,y)} = \dblint{\mathbb{R}^2}{(a^{(1)}(y))^{1/2}Tu(x,y) \cdot  \nabla_x \phi(x,y)}   = 0.
$$
Hence the proposition follows by the virtue of the fact that $C^{\infty}_0(\mathbb{R}^2;V)$ is dense\footnote{The density argument goes along the lines of showing that smooth approximations, via mollification, of $u \in U$ also belong to $U$  by utilising \eqref{emspaceU} and the `self-adjoint' properties of the mollifier. For full details see \cite{VPSIVK}.} in $U$.

It remains to show \eqref{eq:emgenfluxprop1}. By \eqref{emdegentensor1} and \eqref{eq:pdhom10} a  fixed $\psi \in W$ satifies
\begin{multline}
\label{emextraeq11}
\int_{Q_1} \left( \psi_{11} + \psi_{22} \right) \left( \phi_{1,1} + \phi_{2,2} \right) + \left( \psi_{12} - \psi_{21} \right) \left( \phi_{1,2} - \phi_{2,1} \right) \ \mathrm{d}y = 0, \\ \forall \phi \in \left[C^\infty_{\#}(Q)\right]^2.
\end{multline}
For fixed $\varphi \in C^{\infty}_{\#}(Q)$ let $u \in C^{\infty}_{\#} (Q)$ be a solution to 
$$
 \Delta u = \varphi - \frac{1}{\vert Q_0 \vert} \chi_0(y) \mv{\varphi}.
$$
Then, by choosing  $\phi = ( u_{,1} , u_{,2} )$ to be a  test functions in \eqref{emextraeq11} gives
$$
0 = \int_{Q_1} \left( \psi_{11} + \psi_{22} \right) \Delta u \ \mathrm{d}y = \int_{Q_1} \left( \psi_{11} + \psi_{22} \right) \varphi \ \mathrm{d}y,
$$
i.e. $\psi_{11} + \psi_{22} \equiv 0$ for a.e. $y \in Q_1$. Similarly, choosing $\phi = ( -u_{,2} , u_{,1} )$ to be the test function in \eqref{emextraeq11} gives
$$
0 = \int_{Q_1} \left( \psi_{12} - \psi_{21} \right) \Delta u \ \mathrm{d}y = \int_{Q_1} \left( \psi_{12} - \psi_{21} \right) \varphi \ \mathrm{d}y,
$$
i.e. $\psi_{12} - \psi_{21} \equiv 0$ for a.e. $y \in Q_1$.
Hence, \eqref{eq:emgenfluxprop1} holds.
\end{proof}
One is now able to pass to the limit in problem \eqref{finalproblem} using Lemma \ref{lem:pdhom1} and  Proposition \ref{prop:emgenflux} to find the following result.
\begin{thm}
\label{homlimitcase1}
Consider $f^\ep \in [L^2(\RR^2)]^2$, $f \in [L^2(\RR^2 \times Q)]^2$ such that $f^{\ep} \twoscale f$ as $\ep \rightarrow 0$. For the sequence of solutions $u^\ep$ to \eqref{finalproblem}  there exists $u \in L^2(\RR^2 ; V)$ such that, up to a subsequence, $u_{\ep} \twoscale u$. Furthermore, $u$ is the unique solution of the following equation
\begin{multline}
\label{emhomlimit}
\int_{\RR^2}\int_{Q}   \nabla_y u_1 \cdot \overline{\nabla_y \phi_1} + \gamma^{-1} \Big( {\rm div}_y u \cdot \overline{ {\rm div}_y \phi} + {\rm div}_y  u^\perp \cdot \overline{ {\rm div}_y \phi^\perp} \Big) \\ + \int_{\RR^2}\int_{Q} \epsilon_1 \left( u \cdot \overline{ \phi} + \gamma \chi_0 u_1  \overline{ \phi_1} \right) = \int_{\RR^2}\int_{Q} \epsilon_1 \left( f \cdot \overline{\phi}  + \gamma \chi_0 f_1 \overline{ \phi_1} \right), \quad \forall \phi \in C^\infty_0(\RR^2 ;V),
\end{multline} 
where $\gamma : = \tfrac{\epsilon_0}{\epsilon_1} - 1$.
\end{thm}
\begin{proof}
Suppose $\phi \in C^\infty_0(\RR^2 ; V)$ and set $\phi^\ep (x) = \phi\left( x , \tfrac{x}{\ep} \right)$. Note by \eqref{emextralemeq2}, \eqref{emextralemeq4}, \eqref{eq:pdhom7}, \eqref{eq:pdhom8}, and Proposition \ref{prop:emgenflux} that
\begin{flalign*}
\lim_{\ep \rightarrow 0} \int_{\mathbb{R}^2} & \Big( \chi_1\left(\tfrac{x}{\ep}\right) A^\ep_{1}+ \tfrac{\ep^2}{\epsilon_0 - \epsilon_1 + \ep^2} \chi_0\left(\tfrac{x}{\ep}\right)A^\ep_{0} \Big)  \nabla u^\ep \cdot \overline{\nabla \phi^\ep} \ \mathrm{d}x \\
& = \lim_{\ep \rightarrow 0} \int_{\mathbb{R}^2} \left[ a^{(1)}\left(\tfrac{x}{\ep}\right) + \ep^2 a^{(0)}\left(\tfrac{x}{\ep}\right) \right] \nabla u^\ep \cdot \overline{\nabla \phi^\ep} \ \mathrm{d}x \\
& = \lim_{\ep \rightarrow 0} \int_{\mathbb{R}^2} \left[ a^{(1)}\left(\tfrac{x}{\ep}\right) + \ep^2 a^{(0)}\left(\tfrac{x}{\ep}\right) \right] \nabla u^\ep \cdot \left(  \overline{\nabla_x \phi\left( x , \tfrac{x}{\ep} \right)} + \ep^{-1} \overline{\nabla_y \phi\left( x , \tfrac{x}{\ep} \right)} \right) \ \mathrm{d}x \\
& = \lim_{\ep \rightarrow 0} \int_{\mathbb{R}^2} a^{(1)}\left(\tfrac{x}{\ep}\right) \nabla u^\ep \cdot\overline{\nabla_x \phi\left( x , \tfrac{x}{\ep} \right)}  +  a^{(0)}  \left(\tfrac{x}{\ep}\right) \ep\nabla u^\ep \cdot \left( \ep \overline{\nabla_x \phi\left( x , \tfrac{x}{\ep} \right)} + \overline{\nabla_y \phi\left( x , \tfrac{x}{\ep} \right)} \right) \ \mathrm{d}x \\
& = \dblint{\mathbb{R}^2}{ a^{(0)}(y) \nabla_y u(x,y) \cdot \overline{\nabla_y \phi(x,y)}}. 
\end{flalign*}
Therefore, passing to the limit $\ep \rightarrow 0$ in \eqref{finalproblem} for test functions of the form $\phi^\ep$ gives
\begin{multline}
\label{fpthproofe1}
 \dblint{\mathbb{R}^2}{ a^{(0)}(y) \nabla_y u(x,y) \cdot \overline{\nabla_y \phi(x,y)}} +  \dblint{\mathbb{R}^2} {\rho (y) u(x,y) \cdot \overline{ \phi(x,y)}} \\ =  \dblint{\mathbb{R}^2}{ \rho(y) f(x,y) \cdot \overline{ \phi(x,y)}}, \qquad \forall \phi \in C^\infty_0(\RR^2 ; V).
\end{multline}
It remains to show \eqref{fpthproofe1} is equivalent to \eqref{emhomlimit} and this simply seen by recalling the definition of $\rho$, $a^{(0)}$ and the space $V$, see  \eqref{rho}, \eqref{emdegentensor2} and \eqref{emspaceV} respectively.

Existence and uniqueness of the solution to \eqref{fpthproofe1} follows by noting that $C^\infty_0(\RR^2 ; V)$ is dense in $L^2(\RR^2 ; V)$  with respect to the standard norm on $L^2(\RR^2 ; [H^1_{\#}(Q)]^2)$ and recalling \eqref{normofV}.
\end{proof}
Recalling the operator $\A(\theta)$ for $\theta =0$, introduced in Section \ref{secblochlim}, it is clear that the spectrum associated to problem \eqref{fpthproofe1} coincides with $\sigma(0)$, the spectrum of $\A(0)$. In particular, we have shown by Theorem \ref{homlimitcase1} and strong resolvent two-scale convergence  that
\begin{equation}
\label{limzerospec}
\lim_{\ep \rightarrow 0}\sigma_\ep \supset \sigma (0).
\end{equation}

 As one expects the continuous spectrum of $\A^\ep$ to be non-empty, the spectral completeness $\lim_{\ep \rightarrow 0}\sigma_\ep \subset   \sigma (0)$ is not be expected. In this sense, the  limit operator appears to be incomplete. This incompleteness is a consequence of the fact that, for a fixed $N \in \NN^2$,  we can always find a subsequence of  the solutions $u^\ep$ to \eqref{finalproblem} which two-scale converge to a non-trivial $NQ$-periodic function. Here $NQ = [0, N_1) \times [0, N_2)$ and  non-triviality  is taken to mean the $NQ$ periodic function is not $Q$ periodic for $N\neq (1,1)$. 

\subsection{Quasiperiodic homogenisation}
\label{sec:quasiotwoscale}
At the end of the previous section we argued that fixing our periodic reference cell to be $Q$ does not appear to be sufficient for finding the limiting spectrum $\lim_{\ep \rightarrow 0}\sigma_\ep$. This argument will be made precise here. To this end,  we will shall consider the strong $NQ$-periodic two-scale resolvent limit of $\A^\ep$ corresponding by \eqref{finalproblem}. Here $NQ : = [0, N_1) \times [0, N_2)$ for a given  multi-index $N \in \NN^2$.

\noindent Let us begin by reviewing the definition of two-scale convergence (\cite{nguetseng, allaire, zhikov1}), with respect to the  periodic reference cell $NQ$.

\begin{defn}
\label{blochtwoscale}
Let $u_\ep$ be a bounded sequence in $L^2(\RR^2)$. We say $u_\ep$ two-scale converges to $u(x,y) \in L^2(\RR^2 \times NQ)$  if 
\begin{multline*}
\int_{\RR^2}{u_{\ep}(x) \varphi(x)\phi\left(\tfrac{x}{\ep}\right)} \ \mathrm{d}x \longrightarrow \frac{1}{N_1 N_2}\dblintn{\RR^2}{u(x,y)\varphi(x)\phi(y)}, \\
\forall \varphi \in C^{\infty}_0(\RR^2), \forall \phi \in C^{\infty}_{\#}(NQ).
\end{multline*}
\end{defn}
As the size of the periodicity cell does not alter the arguments detailed in Section \ref{sec:PChom}, the analogies of Lemma \ref{lem:pdhom1}, Proposition \ref{prop:emgenflux} and Theorem \ref{homlimitcase1} hold when the periodic reference cell $Q$ is replaced by $NQ$. In particular, we have demonstrated that the following result holds
\begin{thm}
\label{NQthm}
For each $N \in \NN^2$, let $f^\ep \in [L^2(\RR^2)]^2$ two-scale converges to $f \in [L^2(\RR^2 \times NQ)]^2$ as $\ep \rightarrow 0$. Then, the sequence  of solutions $u^\ep$ to \eqref{finalproblem} has a two-scale convergent subsequence to some $u \in L^2(\RR^2 ; V_N)$, where
\begin{equation}
\label{spaceVn}
V_N := \left\{ v \in [H^1_{\#}(NQ)]^2 : \text{$\div {v} = 0$ and $\div{v^{\perp}} = 0$ in $F_1 \cap NQ$}  \right\}.
\end{equation}
Here $F_1$ is the $Q$-periodic extension of the set $Q_1$ to the whole space. Furthermore, $u \in V_N$ satisfies
\begin{multline}
\label{eq:pcbloch1}
\int_{\RR^2}\int_{NQ} \nabla_y u_1 \cdot \overline{\nabla_y \phi_1} + \gamma^{-1}  \Big( {\rm div}_y u \cdot \overline{ {\rm div}_y \phi} + {\rm div}_y  u^\perp \cdot \overline{ {\rm div}_y \phi^\perp}  \Big) \ \mathrm{d}y\mathrm{d}x
\\ + \int_{\RR^2}\int_{NQ} \epsilon_1 \left( u \cdot \overline{\phi}  + \gamma \chi_0 u_1 \overline{ \phi_1} \right) = \int_{\RR^2}\int_{NQ} \epsilon_1 \left( f \cdot \overline{\phi}  + \gamma \chi_0 f_1 \overline{ \phi_1} \right), \quad \forall \phi \in C^\infty_0(\RR^2 ; V_N).
\end{multline}
\end{thm}
Denoting by $\H_N$ the closure of $V_N$ in $[L^2_\rho(NQ)]^2$ we introduce the non-negative closed quadratic form
 $\a_N : \H_N \times \H_N \rightarrow \mathbb{R}$, defined by
 $$
 \a_N(u,v) : = \int_{NQ} \nabla u_1 \cdot \overline{\nabla \phi_1} + \gamma^{-1} \left( \div{ u} \cdot \overline{  \div \phi} + \div  u^\perp \cdot \overline{ \div \phi^\perp} \right).
 $$
We find, by the strong two-scale resolvent convergence, that 
\begin{equation}
\label{eq:pcbloch2}
\lim_{\ep \rightarrow 0}\sigma_\ep \supset \sigma(\A_N),
\end{equation}
where $\sigma(\A_N)$ is the discrete spectrum of the self-adjoint operator $\A_N$ generated by the form $\a_N$. In particular, \eqref{limzerospec} and \eqref{eq:pcbloch2} imply that
\begin{equation}
\label{limspec01}
\lim_{\ep \rightarrow 0}\sigma_\ep \supset \overline{\bigcup_{N \in \NN^2} \sigma(\A_N)}.
\end{equation}
Note that, to prove \eqref{semicontspect} it remains to show the following result.
\begin{lem}
\label{homlimitbloch1}
For $\theta \in [ 0, 1)^{2}$, let $\sigma(\theta)$ be the spectrum of the operator $\A(\theta)$ introduced in Section \ref{secblochlim}. Then, the following inclusion holds
\begin{equation}
\label{limspec02}
\overline{\bigcup_{N \in \NN^2} \sigma(\A_N)} \supset \bigcup_{\theta \in [0,1)^2} \sigma (\theta).
\end{equation}
\end{lem}
\begin{proof}[Proof of Lemma \ref{homlimitbloch1}]
In Section \ref{secblochlim} we showed that the spectra $\sigma(\theta)$ of $\A(\theta)$ are discrete and that the eigenvalues $\lambda_k(\theta)$ are continuous with respect to $\theta$. The continuity of the eigenvalues, and the fact that the rationals are dense in $\RR$, implies that
\begin{equation*}
\label{hombloch2}
\overline{\bigcup_{N \in \NN^2} \hspace{2pt} \bigcup_{\substack{ j \in \NN^2 \\0 \le  j  \le N-\textbf{1}}} \sigma\left( \tfrac{j}{N}\right)} = \bigcup_{\theta \in \left[0 , 1 \right)^2} \sigma{(\theta)}.
\end{equation*}
Henceforth, for fixed multi-indices $j=(j_1,j_2)$, $N=(N_1,N_2)$, $\tfrac{j}{N} : = \left( \tfrac{j_1}{N_1}, \tfrac{j_2}{N_2} \right)$ and we denote $ 0 \le j \le N -1$ to mean $0 \le j_i \le N_i - 1, i =1,2.$

In particular, to prove \eqref{limspec02} it is sufficient to show, for fixed $N \in \NN^2$ and $j \in \NN^2$ with $ 0 \le j \le N-1$, that
\begin{equation}
\label{spe21}
\sigma\left(\A_N\right) \supset \sigma \left(\tfrac{j}{N}\right). 
\end{equation}

Suppose we take such an $N \in \NN^2$ and $j \in \NN^2$. Consider $\left(\lambda(\tfrac{j}{N}),w\right)$,  an eigenvalue-eigenfunction pair of $\A\left(\tfrac{j}{N}\right)$, i.e. $w \in V\left(\tfrac{j}{N}\right)$ and
\begin{multline}
\label{eq:pcbloch1.1}
\int_{Q}  \nabla w_1 \cdot \overline{\nabla \phi_1}  + \gamma^{-1} \left( \div w \cdot \overline{ \div \phi} + \div  w^\perp \cdot \overline{ \div \phi^\perp} \right) 
= \lambda\left(\tfrac{j}{N}\right) \int_{Q}\epsilon_1 \left( w \cdot \overline{\phi}  + \gamma \chi_0 w_1 \overline{ \phi_1} \right) \\ \quad \forall \phi \in V\left(\tfrac{j}{N}\right).
\end{multline}
Recall, by \eqref{rho} and \eqref{emdegentensor2}, that \eqref{eq:pcbloch1.1} is equivalent to 
\begin{equation}
\label{eq:pcbloch1.12}
\int_{Q} a^{(0)} \nabla w \cdot \overline{\nabla \phi} = \lambda(\tfrac{j}{N}) \int_{Q} \rho w \cdot \overline{\phi},  \quad \forall \phi \in V\left(\tfrac{j}{N}\right).
\end{equation}

As any function  $u \in H^1_{j / N}(Q)$  has the representation $u(y) = e^{{\rm i} 2\pi \left(j/N\right) \cdot y} \widetilde{u}(y)$ for some $\widetilde{u} \in H^1_{\#}(Q)$, we can extend $u$ to belong to $H^1_{\#}(NQ)$ by $Q$-periodically extending $\widetilde{u}$ to $NQ$. Therefore, we perform such an extension on the solution $w$ to \eqref{eq:pcbloch1.12}, and can directly shown that $w \in V_N$, see \eqref{spaceVn}. Furthermore, for fixed $\varphi \in C^\infty_{\#}(Q) \cap V_N$, for the periodic second-order tensor $a^{(0)}$, one has
\begin{flalign}
\int_{NQ} a^{(0)} \nabla  w \cdot  \overline{\nabla \varphi} \nonumber & = \sum_{n_1 = 0}^{N_1 - 1}\sum_{n_2 = 0}^{N_2 - 1} \int_{Q+n} a^{(0)}(y ) \nabla_y w(y) \cdot \overline{\nabla_y\varphi(y)}\ \mathrm{d}y \nonumber \\
& = \sum_{n_1 = 0}^{N_1 - 1}\sum_{n_2 = 0}^{N_2 - 1} \int_{Q} a^{(0)}(y + n ) \nabla_y w(y + n) \cdot \overline{\nabla_y\varphi(y + n)}\ \mathrm{d}y \nonumber \\
&  = \sum_{n_1 = 0}^{N_1 - 1}\sum_{n_2 = 0}^{N_2 - 1} \int_{Q} a^{(0)}(y) \nabla_y e^{  2 \pi {\rm i} (j/N)  \cdot n }w(y ) \cdot \overline{\nabla_y\varphi(y + n)}\ \mathrm{d}y  \nonumber\\
& =  \int_{Q} a^{(0)} \nabla w \cdot \overline{\nabla \phi}, \label{eq:pcbloch1.2}
\end{flalign}
where
\begin{equation}
\label{testbloch}
\phi(y) : = \sum_{n_1 = 0}^{N_1 - 1}\sum_{n_2 = 0}^{N_2 - 1} e^{- 2 \pi {\rm i} ( j / N)\cdot n} \varphi( y + n).
\end{equation}
 Direct calculation shows $\phi(y) \in C^\infty(Q) \cap [H^1_{j / N}(Q)]^2$: Indeed, in the $x_1$ direction,
\begin{flalign*}
\phi(y + e_1 ) & = \sum_{n_1 = 0}^{N_1 - 1}\sum_{n_2 = 0}^{N_2 - 1} e^{\left(-  2 \pi {\rm i} (n_1j_1/N_1 + n_2j_2/N_2 )\right)} \varphi( y + n + e_1) \\
& = e^{  2 \pi {\rm i} (j_1/N_1 )}\sum_{m_1 = 1}^{N_1 }\sum_{m_2 = 0}^{N_2 -1 } e^{\left(- 2 \pi {\rm i} (m_1j_1/N_1 + m_2j_2/N_2 )\right)} \varphi( y + m) \\
& = e^{ 2 \pi{\rm i}  (j_1/N_1  )}\phi(y),
\end{flalign*}
with similar calculations holding in the $x_2$ direction and for the first-order derivatives of $\phi$. Moreover, $\phi(y) \in V\left(\tfrac{j}{N}\right)$ since $ a^{(1)}\nabla_y \varphi(y) = 0$ in $F_1 \cap NQ$ because $\varphi \in V_N$.

The equations \eqref{eq:pcbloch1.12}, \eqref{eq:pcbloch1.2} and \eqref{testbloch} imply that
\begin{flalign*}
\int_{NQ} a^{(0)} \nabla w \cdot \overline{\nabla \varphi} & = \int_{Q} a^{(0)} \nabla w \cdot \overline{\nabla \phi}  = \lambda(\tfrac{j}{N}) \int_{Q} \rho w \cdot \overline{\phi} \ \mathrm{d}y \\
& = \lambda(\tfrac{j}{N}) \sum_{n_1 = 0}^{N_1 - 1}\sum_{n_2 = 0}^{N_2 - 1} \int_{Q} \rho(y) e^{  2 \pi {\rm i}(j / N) \cdot n} w(y ) \cdot \overline{(\varphi(y + n))}\ \mathrm{d}y \\
& = \lambda(\tfrac{j}{N}) \sum_{n_1 = 0}^{N_1 - 1}\sum_{n_2 = 0}^{N_2 - 1} \int_{Q} \rho(y + n) w(y +n) \cdot \overline{(\varphi(y + n))}\ \mathrm{d}y \\
& = \lambda(\tfrac{j}{N}) \int_{NQ} \rho w \cdot \overline{\varphi}.
\end{flalign*}
This implies, by the arbitrariness of $\varphi \in C^\infty_{\#}(Q) \cap V_N$, that  $\lambda(\tfrac{j}{N})$ belongs to the spectrum of $\A_N$, i.e. \eqref{spe21} holds.

\end{proof}

\section{Completeness of spectrum}
In this section we shall prove that the inclusion
\begin{equation}
\label{limspe23}
\lim_{\ep \rightarrow 0} \sigma_\ep \subset \bigcup_{\theta \in \left[ 0 , 1 \right)^2} \sigma{(\theta)} = : \sigma_0.
\end{equation}
holds. More precisely, we shall  show that if $\lambda_\ep \in \sigma_\ep$, such that $\lambda_\ep \rightarrow \lambda$ as $\ep \rightarrow 0$, then $\lambda \in \sigma{(\theta)}$ for some $\theta \in [0,1)^2$. Along with \eqref{semicontspect}, the above inclusion proves Theorem \ref{mainthm1}.

For fixed $\lambda_\ep \in \sigma_\ep$, we know by the Floquet-Bloch decomposition, that there exist $\eta^\ep \in \ep^{-1} [0,1)^2,$ $w^\ep \in [H^1_{\eta^\ep}(\ep Q)]^2$, such that 
\begin{multline}
\label{eq:spcom1}
\int_{\ep Q} \Big( \chi_1\left(\tfrac{x}{\ep}\right) A^\ep_{1} + \tfrac{\ep^2}{\epsilon_0 - \epsilon_1 + \ep^2} \chi_0\left(\tfrac{x}{\ep}\right)A^\ep_{0} \Big) \nabla w^\ep \cdot \overline{\nabla \phi} \ \mathrm{d}x  = \lambda_\ep \int_{\ep Q} \rho\left(\tfrac{x}{\ep}\right) w^{\ep} \cdot \overline{\phi} \ \mathrm{d}x, \\ \qquad \forall \phi \in [C^\infty_{\eta^\ep}(\ep Q)]^2.
\end{multline} 
Upon, choosing in \eqref{eq:spcom1} the rescalings $y = \tfrac{x}{\ep}$ and $\theta_\ep = \ep \eta_\ep \in [0,1)^2$, then $u^\ep(y) : = w^\ep(\ep y)$ belongs to $[H^1_{\theta^\ep}(Q)]^2$ and solves 
\begin{multline}
\label{eq:spcom2}
\int_{Q} \Big(  \tfrac{1}{\ep^2}\chi_1\left( y \right) A^\ep_{1} + \tfrac{1}{\epsilon_0 - \epsilon_1 + \ep^2} \chi_0\left(y \right)A^\ep_{0} \Big) \nabla u^\ep \cdot \overline{\nabla \phi} \ \mathrm{d}x  = \lambda_\ep \int_{Q} \rho\left(y \right) u^{\ep} \overline{\phi} \ \mathrm{d}y, \\ \qquad \forall \phi \in [C^\infty_{\theta^\ep}(Q)]^2.
\end{multline}
Under the assumption that $\lambda_\ep \rightarrow \lambda$, we have, up to a subsequence, $\theta^\ep \rightarrow \theta$ as $\ep \rightarrow 0$. Therefore, to prove \eqref{limspe23} it is sufficient to show that the corresponding, appropriately normalised, sequence $u^\ep$ of solutions to \eqref{eq:spcom2} converges strongly in $[L^2(Q)]^2$ to some non-trivial function $u \in V(\theta)$ which solves
\begin{multline*}
\int_{Q}  \nabla u_1 \cdot \overline{\nabla \phi_1}  + \gamma^{-1} \left( \div u \cdot \overline{ \div \phi} + \div  u^\perp \cdot \overline{ \div \phi^\perp} \right) 
= \lambda \int_{Q}\epsilon_1 \left( u \cdot \overline{\phi}  + \gamma \chi_0 u_1 \overline{ \phi_1} \right), \\ \forall \phi \in V(\theta),
\end{multline*}
or, equivalently (see \eqref{emdegentensor2})
$$
\int_{Q} a^{(0)} \nabla u \cdot \overline{\nabla \phi} = \lambda \int_{Q} \rho u \cdot \overline{\phi}, \quad \forall \phi \in V(\theta).
$$
Indeed, if true this implies $\lambda \in \sigma{(\theta)}$. This is the main spectral completeness result.
\begin{thm}
\label{thm:spcom1}
For  $\lambda_\ep \in \sigma_\ep$, let $\theta^\ep \in [0, 1)^2$, $u^\ep \in [H^1_{\theta^\ep}(Q)]^2$ with $\norm{u^\ep}_{[L^2(Q)]^2} = 1$ satisfying
\begin{multline}
\int_{Q} \Big(  \tfrac{1}{\ep^2}\chi_1\left( y \right) A^\ep_{1} + \tfrac{1}{\epsilon_0 - \epsilon_1 + \ep^2} \chi_0\left(y \right)A^\ep_{0} \Big)\nabla u^\ep \cdot \overline{\nabla \phi} \ \mathrm{d}y = \lambda_\ep \int_{Q} \rho(y) u^\ep \cdot \overline{\phi} \ \mathrm{d}y, \\ \forall  \phi \in [H^1_{\theta^{\ep}}(Q)]^2. \label{eq:spcom3}
\end{multline}
Assume $\lambda_\ep \rightarrow \lambda$, $\theta^\ep \rightarrow \theta$ as $\ep \rightarrow 0$. Then, there exists a unique $u \in V(\theta)$, $u \neq 0$, such that 

\begin{equation}
\label{eq:spcom4}
\int_{Q} a^{(0)} \nabla u \cdot \overline{\nabla \phi}  = \lambda \int_{Q} \rho u \cdot \overline{\phi}, \quad \forall \phi \in V(\theta).
\end{equation}
\end{thm}

\begin{proof}
From the assumptions of the theorem, \eqref{emextraeq1}, \eqref{emtensorprop}, \eqref{emdegentensor1} and by choosing test functions $\phi = u^\ep$ in \eqref{eq:spcom3}, we have the following bounds: there exists $C$ independent of $\ep$ such that
\begin{align*}
\norm{u^\ep}_{[L^2(Q)]^2} & =1 & \norm{\nabla u^\ep}_{[L^2(Q)]^{2\times2}} &\le C & \norm{ \left( a^{(1)} \right)^{1/2} \nabla u^\ep}_{[L^2(Q)]^{2\times 2}} &\le \ep^2 C.
\end{align*}
This implies $u^\ep$ converges, up to a subsequence, to some $u$ weakly in $[H^1(Q)]^2$ and therefore strongly in $[L^2(Q)]^2$ as $\ep \rightarrow 0$. Furthermore, $\left( a^{(1)} \right)^{1/2} \nabla u^\ep \rightarrow 0$ in $[L^2(Q)]^{2\times2}$, which imples  $ \left( a^{(1)} \right)^{1/2} \nabla u^0 =0$; i.e. we have shown that  $u \in V(\theta)$, see \eqref{emspaceVn}. \\

It remains to show that $u$ satisfies \eqref{eq:spcom4}. For any fixed $\phi \in V(\theta)$, by Lemma \ref{lem:spcom1} there exists $\phi^\ep \in V(\theta^\ep)$ such that $\phi^\ep \rightarrow \phi$ strongly in $[H^1(Q)]^2$ as $\ep \rightarrow 0$. Then, by taking $\phi^\ep$ as a test function in \eqref{eq:spcom3}, recalling $a^{(1)}(y) \nabla u^\ep \cdot \overline{\nabla \phi^\ep} = 0$ and by using the appropriate analogue of \eqref{emextralemeq4}, gives
$$
\lim_{\ep \rightarrow 0}\int_{Q} \Big(  \tfrac{1}{\ep^2}\chi_1\left( y \right) A^\ep_{1} + \tfrac{1}{\epsilon_0 - \epsilon_1 + \ep^2} \chi_0\left(y \right)A^\ep_{0} \Big)\nabla u^\ep \cdot \overline{\nabla \phi^\ep} \ \mathrm{d}y  = \lim_{\ep \rightarrow 0} \int_{Q} a^{(0)} \nabla u^\ep \cdot \overline{\nabla \phi^\ep}.
$$
Therefore, since $\phi^\ep \rightarrow \phi$, $u^\ep \rightharpoonup u$ in $[H^1(Q)]^2$ as $\ep \rightarrow 0$, we find that
$$
\lim_{\ep \rightarrow 0}\int_{Q} \Big(  \tfrac{1}{\ep^2}\chi_1\left( y \right) A^\ep_{1} + \tfrac{1}{\epsilon_0 - \epsilon_1 + \ep^2} \chi_0\left(y \right)A^\ep_{0} \Big)\nabla u^\ep \cdot \overline{\nabla \phi^\ep} \ \mathrm{d}y =\int_{Q} a^{(0)} \nabla u \cdot \overline{\nabla \phi}.
$$
Additionaly, one can easily see that
$$
\lim_{\ep \rightarrow 0}\lambda_\ep \int_{Q} \rho u^\ep \cdot \overline{\phi^\ep} = \lambda \int_{Q} \rho u \cdot \overline{\phi}.
$$
This implies $u$ solves equation \eqref{eq:spcom4} and, since $u^\ep$converges to $u$ strongly in $[L^2(Q)]^2$ one has $\norm{u}_{[L^2(Q)]^2} = 1$, i.e. $u^0 \neq 0$. This completes the proof of the theorem.
\end{proof}

\section{Photonic crystals which exhibit gaps in the limit spectrum $\sigma_0$}
\label{examples}
In this section we will demonstrate gaps in the limit spectrum 
$$
\sigma_0 = \bigcup_{\theta \in [0,1)^2} \sigma(\theta)
$$
for specific photonic crystals. The limit spectrum $\sigma_0$ will have a gap if any two adjacent bands do not overlap,  that is, if for some $i \in \mathbb{N}$, 
$$\max_{\theta \in [0,1)^2} \lambda_i (\theta) < \min_{\theta \in [0,1)^2} \lambda_{i+1} (\theta).$$
Therefore, for each $\theta \in [0,1)^2$, we are interested in studying the following eigenvalue problem: Find $u \in V(\theta)$ such that
\begin{equation}
\label{homlimitcase1.1}
\int_{Q} \nabla u_1 \cdot \overline{\nabla \phi_1} + \gamma^{-1} \left( \div u \cdot \overline{\div \phi} + \div u^\perp \cdot \overline{\div \phi^\perp} \right) =  \lambda(\theta) \int_{Q} \epsilon_1 \left(  u \cdot \overline{\phi} +  \gamma \chi_0 u_1 \overline{\phi_1} \right),  \ \ \forall \phi \in V(\theta).
\end{equation} 
Here, we recall that $\gamma : = \tfrac{\epsilon_0}{\epsilon_1} - 1$ and
\begin{equation}
\label{limexV} 
\quad V(\theta) : = \left\{ u \in [H^1_{\theta}(Q)]^2 : \ \text{ $\div u = 0$ and  $\div u^{\perp} = 0$ in  $Q_1$}\right\}.
\end{equation}
Let us now consider some particular problems.
\subsection{Normal incident wave on a one-dimensional photonic crystal}
We consider a one-dimensional multilayer photonic crystal. Namely, consider a slab like inclusion $Q_1 = [a,b]\times[0, 1)$, $0 <a <b < 1$. Seek a solution of the form $u(x,y) = v(x,y_1)e^{2\pi{\rm i}ky_2}$. Then \eqref{limexV} implies that $u \in V(\theta)$ for $\theta=(\theta_1, \theta_2)$, if, and only if, $v$ satisfies 
\begin{flalign}
\pd{v_1}{y_1} + 2\pi{\rm i}k v_2 & = 0, \quad \text{in $Q_1$} \label{1dslabeq1}\\
2\pi{\rm i}kv_1 - \pd{v_2}{y_1} & = 0, \quad \text{in $Q_1$}. \label{1dslabeq2}
\end{flalign}
That yields $v(x,y_1) = a_1(x)\big( {-2\pi{\rm i}e^{ky_1}}, {e^{ky_1}} \big) + a_2(x)\big( 2\pi{\rm i}e^{-ky_1}, {e^{-ky_1}} \big)$ in $Q_1$. Heuristically, we expect that propagation is more likely to be forbidden in direction of greatest the dielectric discontinuity, therefore let us consider waves that propagate in the direction with the greatest variation in electric permittivity, i.e set $k=0$ and $\theta_2=0$. Equations \eqref{1dslabeq1}-\eqref{1dslabeq2} imply that $u \in V(\theta)$ if, and only if, $v$ is $\theta$-quasi periodic and is a constant vector in $Q_1$, i.e. $v \in [H^1_{\theta_1}[0,1)]^2,$ and $v$ is a constant vector in $[a,b]$.

We now wish to study problem \eqref{homlimitcase1.1}, which takes the following form: For fixed $\theta_1 \in [0,1)$, find  $\lambda(\theta_1)$ and $v=(v_1,v_2) \in [H^1_{\theta_1}[0,1)]^2,$ such that $v$ is a constant vector in $[a,b]$ and
\begin{multline}
\label{homlimitcase1.2}
\int_{[0,a) \cup (b,1]} \left(\begin{matrix} 2v_1^{'} \\ v_2^{'} \end{matrix} \right) \cdot \left( \begin{matrix} \overline{\phi_1^{'}} \\ \overline{\phi_2^{'}} \end{matrix} \right) =  \lambda(\theta_1) \int_{0}^{1}  \rho u \cdot \overline{\phi}, \quad \forall \phi \in [H^1_{\theta_1}[0,1)]^2, \text{ and $\phi$ is constant in $[a,b]$}.
\end{multline} 
Here
\begin{equation*}
\rho(y) = \chi_1(y)\left( \begin{array}{cc} 1 & 0 \\ 0 & 1\end{array} \right) + \chi_0(y) \left( \begin{array}{cc} 2 & 0 \\ 0 & 1\end{array} \right). 
\end{equation*}
(We have chosen $\varepsilon_0=2$, $\varepsilon_1=1$ for simplicity). Notice that, since $ \phi \in C^\infty_0([0,1]\backslash [a,b])$ is an admissible test function, $\eqref{homlimitcase1.2}$ implies that
\begin{equation}
\label{slab1}
- v^{''}(y)  = \lambda(\theta_1) v(y) \qquad y\in [0,a) \cup (b,1].
\end{equation}
Furthermore, since $v(y) =C \in \mathbb{C}^2$ in $[a,b]$, integrating by parts in \eqref{homlimitcase1.2}, and using \eqref{slab1}, gives for fixed $\phi$
\begin{equation}
\label{homlimitcase1.3}
 \lambda(\theta_1) \int_{a}^{b}  C \cdot \overline{\phi} \ \mathrm{d}y = 2 v_{1}^{'}(y) \overline{\phi_1}(y)\vert_{y=b}^{a} + v_{2}^{'}(y) \overline{\phi_2}(y) \vert_{y=b}^{a}.
\end{equation} 
The space of admissible test functions is two-dimensional and spanned by functions of the form $(1,0)^T$ and $(0,1)^T$ in $[a,b]$. Therefore \eqref{homlimitcase1.3} can be reduced to the following algebraic system:
\begin{align*}
\lambda(\theta_1)( b - a)C_1 & = 2v_{1}^{'}(a) - 2v_{1}^{'}(b), \\
\lambda(\theta_1)( b - a)C_2 & = v_{2}^{'}(a) - v_{2}^{'}(b).
\end{align*}  
Taking all of this into consideration, we see that solving \eqref{homlimitcase1.2} is equivalent to simultaneously finding $u$ and $C$ such that
\begin{gather}
\begin{aligned}
- v^{''}(y)  = \lambda(\theta_1) v(y) & & y\in [0,a) \cup (b,1], \label{pcslabeq1}\end{aligned} \\
v(1) = e^{ 2 \pi{\rm i} \theta_1} v(0), \qquad v'(1) = e^{2 \pi{\rm i} \theta_1} v'(0), \label{pcslabeq2}
\end{gather}
with the following interface conditions
\begin{align}
v(a)=v(b)&=C \label{pcslabeq3}\\
\lambda(\theta_1)( b - a)C_1 & = 2v_{1}^{'}(a) -2v_{1}^{'}(b), \label{pcslabeq4}\\
\lambda(\theta_1)( b - a)C_2 & = v_{2}^{'}(a) - v_{2}^{'}(b). \label{pcslabeq5}
\end{align} 
Seeking solutions to \eqref{pcslabeq1}-\eqref{pcslabeq5} of the form 
\begin{align*}
v(y) = A^1 e^{{\rm i} \sqrt{\lambda(\theta_1)}\theta_1 y} + A^2e^{ - {\rm i} \sqrt{\lambda(\theta_1)}\theta_1 y} & & y \in [0,a), \\
v(y) = B^1 e^{{\rm i} \sqrt{\lambda(\theta_1)}\theta_1 y} + B^2e^{ - {\rm i} \sqrt{\lambda(\theta_1)}\theta_1 y} & & y \in (b,1],
\end{align*}
for some $A^1=(A^1_1,A^1_2)^T$, $A^2=(A^2_1,A^2_2)^T$, $B^1=(B^1_1,B^1_2)^T$, $B^2=(B^2_1,B^2_2)^T$ to be determined, we arrive at a linear system that can be represented as follows: find $X = (C, A^1, A^2, B^1, B^2)^T$ such that
$$
M(\lambda(\theta_1), \theta_1) X = 0,
$$
where $M(\lambda(\theta_1),\theta_1)$ is the corresponding $10 \times 10$ matrix. To solve this we find the zeros of the function $F: \mathbb{R} \times [0, 1) \rightarrow \mathbb{C}$ defined by $F(x, y) : = \det{M(x,y)}$. The $\lambda = x$,  where $x$ belongs to the set of the zeros of $F$, make up the spectrum of the limit operator, see Figure \ref{fig:homlimitspec}. We find that the limit spectrum of the one-dimensional crystal does indeed have gaps.

\begin{center}
\begin{figure}
		\includegraphics[width=\linewidth]{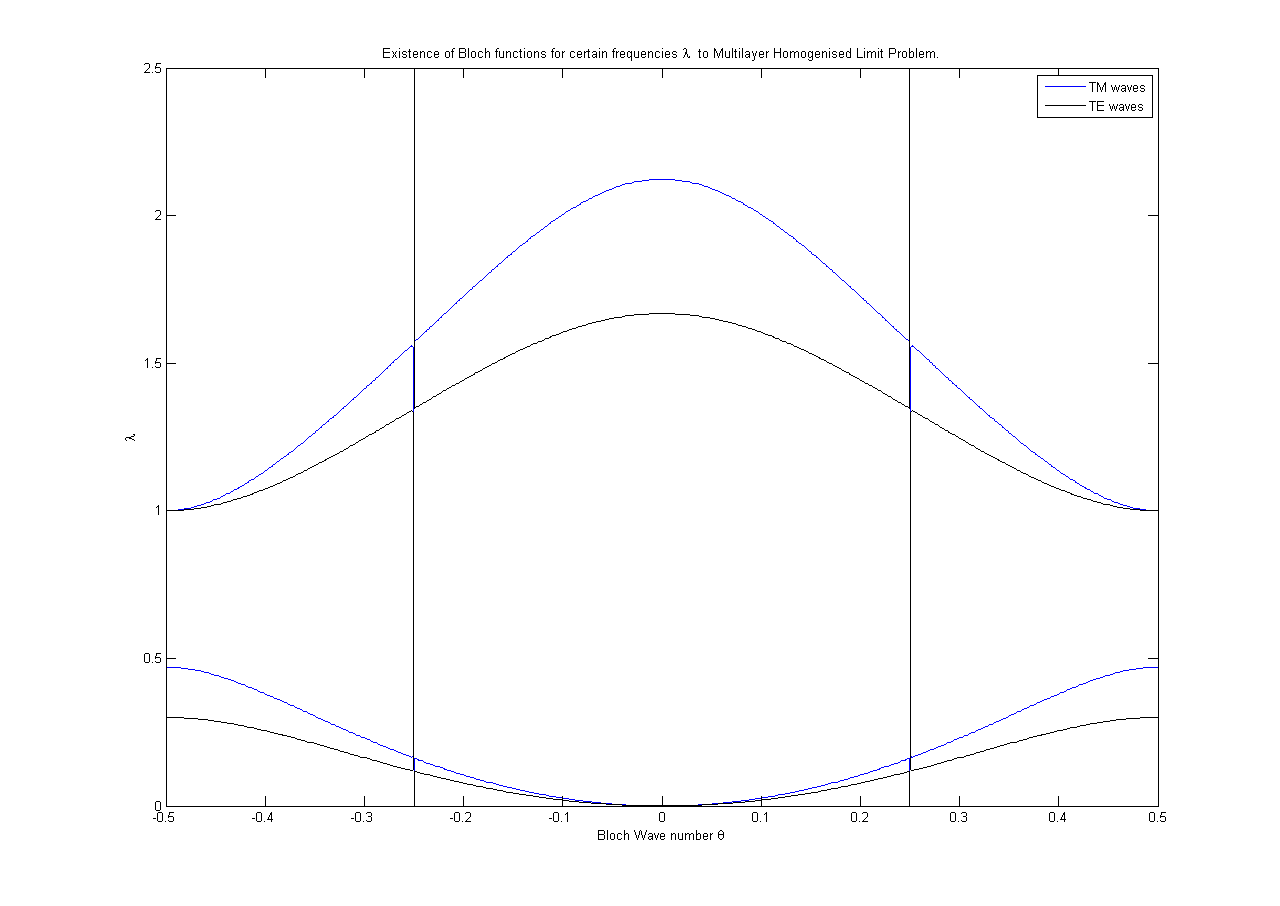}
	\caption{Band gap structure of limit spectrum. Plot was made by using Matlab to find the level curve $F(\lambda, \theta) = 0$.}
	\label{fig:homlimitspec}
\end{figure}
\end{center}
\begin{rem}
In Figure \ref{fig:homlimitspec} we see two curves whose image is a part of the spectrum of our homogenised limit operator. The reason for this lies the symmetry of the domain in $x_1$-direction, waves propagating with wave number $(k_{x_1},k_{x_2})=(0,k)$ are polarised. These are the so-called transverse magnetic (TM) and transverse Electric (TE) polarisations. The presence of such polarisations implies that one could have decomposed problem \eqref{effectivecontrast} into finding independently solutions of the form $u^\ep =( u^\ep_1,0)$ and $u^\ep=(0,u^\ep_2)$. These functions would two-scale converge to functions of the form $(u_1,0)$ and $(0,u_2)$ respectively, therefore leading to an effective TM or TE polarised limit problem. It is to be noted here that the general two-dimensional case has no such polarisations and the solution $u^\ep=(u^\ep_1,u^\ep_2)$ to problem \eqref{effectivecontrast} cannot be decoupled. This makes the two-dimensional case intrinsically more difficult to study. We consider a specific two-dimensional example below, with small inclusions, and prove there exist at least one spectral gap using variational arguments.
\end{rem}

\subsection{ARROW fibres: two-dimensional photonic crystals whose inclusions have vanishing volume faction}
We consider a two-dimensional photonic fibre with small circular inclusions. That is, let $Q_0 = B_\delta(0)$  the ball of radius  $\delta < <1$ centred at the origin. Such a photonic fibre could be created, for example, by drilling a periodic array of cylindrical holes of a very small cross-sectional radius in a dielectric material and then filling the holes with a material that is more optically dense than the background dielectric ( $\varepsilon_0 > \varepsilon_1$ ). (Note in passing that such models are known in physics as ARROW fibres, see e.g. \cite{ARROW}.)

We will show that for sufficiently small $\delta$ that the limit spectrum $\sigma_0$ has at least one gap. To this end, we will show that there exists constants $c_1,c_2 >0$ such that 
\begin{align}
\label{eq:PClimspec1}
\lambda_2 (\theta) &\le - \frac{c_2}{\delta^2 \ln{\delta}},  &  \lambda_3(\theta) &\ge c_1 \delta^{-2}.
\end{align}

\noindent Let us now show that \eqref{eq:PClimspec1} holds. For $\theta = 0$ we see that $\lambda(0) = 0$ is an eigenvalue, with multiplicity 2, for \eqref{homlimitcase1.1}: the orthogonal eigenfunctions, of $\lambda(0) = 0$, are $u_0(y) = (1,0)$ and $v_0(y)=(0,1)$. This implies 
$$
\lambda_1(0) = \lambda_2(0) = 0.
$$
Let us now consider the more interesting, non-trivial case $\theta \neq 0$. By classical variational arguments, it is known that
$$
\lambda_1 (\theta) = \min_{\substack{u \in V(\theta) \\  u \neq 0}} \frac{\int_{Q} a^{(0)} \nabla u \cdot \overline{\nabla u} }{ \int_{Q} \rho u \cdot u}.
$$
By the uniform ellipticity and boundedness of $a^{(0)}$ and $\rho$, it is clear that there exists a constant $C>0$ such that 
$$
\lambda_1 (\theta) \le C \min_{\substack{u \in V(\theta) \\  u \neq 0}} \frac{\int_{Q} \vert \nabla u \vert^2}{\int_{Q} \vert u \vert^2},
$$
We shall now show that the right hand side of the above inequality is bounded from above by $ - \tfrac{c_2}{\delta^2 \ln{\delta}}$. Denote $u = \nabla N$ where $N \in H^2_{\theta}(Q)$ is the unique solution to 
$$
- \Delta N = \chi_0,
$$
where $\chi_0$ is the characteristic function of $Q_0 = B_\delta(0)$. Then
\begin{gather}
\int_{Q} \vert \nabla u \vert^2   \le \int_{Q} \vert \nabla^2 N \vert = \int_Q \vert \Delta N \vert^2  = \delta^2 \vert B_1(0) \vert, \nonumber \\
\intertext{and}
\label{limeqadd3}
\begin{aligned}
\int_{Q} \vert  u \vert^2  & = \int_{Q} \vert \nabla N \vert^2  = - \int_{Q} \Delta N \cdot N =  \int_{ Q_0} N  \ge - c \delta^4 \ln{\delta}
\end{aligned}
\end{gather}
The last inequality is due to subtle technical arguments, see Proposition \ref{prop:pcspectlim1}, Proposition \ref{prop:pcspectlim2} and Corollary \ref{cor:pclimitspec} below. Hence,  $\lambda_1(\theta) \le - \tfrac{c_2}{\delta^2 \ln{\delta}}$. Furthermore, this result also follows for $v = \nabla^{\perp} N$ and, since $u,v$ are readily seen to be orthogonal, the min-max variational principle implies  $\lambda_2(\theta) \le - \tfrac{c_2}{\delta^2 \ln{\delta}}$. \\

\noindent We shall now show that the second inequality in \eqref{eq:PClimspec1} for $\lambda_3(\theta)$ for $\theta \neq 0$ which will imply that it  holds for $\theta = 0$ by continuity of $\lambda(\theta)$.  Suppose $\theta \in (0,1)^2$. By variational principle and ellipticity of $a^{(0)}$ and $\rho$,  
$$
\lambda_3(\theta) \ge C \inf_{\substack{u \perp v \\ u \perp w}} \frac{\int_{Q} \vert \nabla u \vert^2}{\int_{Q} \vert u \vert^2},
$$
for some constant $C>0$ and, for any $v,w \in V(\theta)$ such that $v \perp w$. Here  $\perp$ should be read as  orthogonality in $L^2$. We shall choose $v = \nabla N$, $w= \nabla^{\perp} N$ for $N$ constructed above; clearly $v \perp w$. For fixed $u \in V(\theta)$, by Lemma \ref{lem:spcom2}, $u = \nabla a + \nabla^{\perp} b$ for some $a,b  \in H^2_{\theta}(Q)$ that are harmonic in $Q_1$.  Note that, for $u \perp v$
\begin{align*}
0 & = \int_{Q} u \cdot \overline{v} = \int_{Q} \left( \begin{array}{c} a_{,1} - b_{,2} \\ a_{,2} + b_{,1} \end{array} \right) \cdot \left( \begin{array}{c} \overline{N_{,1}} \\ \overline{N_{,2}} \end{array} \right) =  - \int_{Q} a \overline{\Delta N} =  \int_{Q_0} a,
\end{align*} 
which implies $\int_{Q_0} a =0$. Similarly $u \perp w$ implies $\int_{Q_0} b= 0$. Furthermore, we observe that
\begin{equation}
\label{limeqadd1}
\frac{\int_{Q} \vert \Delta a \vert^2}{\int_{Q} \vert \nabla a \vert^2} = \frac{\int_{Q} \vert \Delta a \vert^2}{ - \int_{Q} \Delta a \cdot \overline{a} } \ge \mu_2 \delta^{-2}, \quad \mu_2>0.
\end{equation}
Indeed, by noticing that
\begin{flalign}
\left\vert \int_Q \Delta a \cdot \overline{a} \right\vert & \le \left( \int_{Q_0} \vert \Delta a \vert^2 \right)^{1/2}\left( \int_{Q_0} \vert a \vert^2 \right)^{1/2} \nonumber \\
 & \le  \left( \int_{Q_0} \vert \Delta a \vert^2 \right)^{1/2} \left( \delta^{-2} \mu_2 \right)^{-1/2} \left( \int_{Q_0} \vert \nabla a \vert^2 \right)^{1/2} \nonumber \\
 & \le \left( \delta^{-2} \mu_2 \right)^{-1/2} \left( \int_{Q_0} \vert \Delta a \vert^2 \right)^{1/2}\left( \int_{Q} \vert \nabla a \vert^2 \right)^{1/2} \nonumber \\ 
 & = \left( \delta^{-2} \mu_2 \right)^{-1/2} \left( \int_{Q_0} \vert \Delta a \vert^2 \right)^{1/2}\left( - \int_{Q}  \Delta a \cdot \overline{a}  \right)^{1/2}. \label{limeqadd2}
\end{flalign}
The second inequality is due to the following Poincar\'{e} type inequality: 
\begin{equation}
\label{limponineq}
\int_{Q_0} \vert a \vert^2  \le  \mu_2^\delta  \int_{Q_0} \vert \nabla a \vert^2,
\end{equation}
where $\mu^\delta_2$ is the first non-zero eigenvalue of the Neumann Laplacian on $Q_0 = B_\delta (0)$. The inequality \eqref{limponineq} can be shown by a simple application of the spectral theory for self-adjoint operators. Furthermore, by a simple rescaling argument, $\mu^\delta_2 = \delta^{-2} \mu_2$, where $\mu_2$ is the first non-zero eigenvalue of the Neumann Laplacian on the unit ball, $B_{1}(0)$. Then \eqref{limeqadd1} immediately follows from \eqref{limeqadd2}.

Similarly
$$
\frac{\int_{Q} \vert \Delta b \vert^2}{\int_{Q} \vert \nabla b \vert^2} \ge \mu_2 \delta^{-2}. 
$$
Furthermore, since $a,b \in H^2_{\theta}(Q)$, then by integration by parts
\begin{flalign*}
\int_Q \nabla a \cdot \overline{\nabla^\perp b}  & = \int_Q -a_{,1}\overline{b_{,2}} + a_{,2}\overline{b_{,1}} = \int_Q a\overline{b_{,21}} - a\overline{b_{,12}} = 0, \\
\int_Q \nabla b \cdot \overline{\nabla^\perp a}  & = \int_Q -b_{,1}\overline{a_{,2}} + b_{,2}\overline{a_{,1}} = \int_Q b\overline{a_{,21}} - b\overline{a_{,12}} = 0.
\end{flalign*}
This implies, for $u = \nabla a + \nabla^\perp b$,
\begin{equation*}
\int_{Q} \vert  u \vert^2 = \int_{Q} \vert \nabla a \vert^2 + \int_{Q} \vert \nabla^\perp b \vert^2 = \int_{Q} \vert \nabla a \vert^2 + \int_{Q} \vert \nabla b \vert^2 .
\end{equation*}
Since, for any $a,b$, by a similar direct inspection,
\begin{align*}
\int_Q \nabla^2 a \cdot \overline{\nabla \left( \nabla^\perp b \right)} = \int_Q \nabla^2 b \cdot \overline{\nabla \left( \nabla^\perp a \right)} = 0,
\end{align*}
we also obtain
\begin{equation*}
\int_{Q} \vert \nabla u \vert^2 = \int_{Q} \vert \Delta a \vert^2 + \int_{Q} \vert \Delta b \vert^2.
\end{equation*}
All of these considerations, and \eqref{limeqadd1}, imply that
$$
\frac{\int_{Q} \vert \nabla u \vert^2}{\int_{Q} \vert  u \vert^2} = \frac{\int_{Q} \vert \Delta a \vert^2}{\int_Q \vert \nabla a \vert^2 + \int_{Q} \vert \nabla b \vert^2} + \frac{\int_{Q} \vert \Delta b \vert^2}{\int_Q \vert \nabla a \vert^2 + \int_{Q} \vert \nabla b \vert^2} \ge \delta^{-2}\mu_2 \frac{\int_{Q} \vert \nabla a \vert^2 + \int_{Q} \vert \nabla b \vert^2}{\int_{Q} \vert \nabla a \vert^2 + \int_{Q} \vert \nabla b \vert^2}.
$$
Hence $\lambda_3(\theta) \ge \delta^{-2} \mu_2.$

Finally, we prove \eqref{limeqadd3}.
\begin{prop}
\label{prop:pcspectlim1}
Let $Q_0 = B_\delta(0) \subset \mathbb{R}^2$ be the open ball of radius $\delta < \tfrac{1}{2}$ centred at the origin. For fixed $\theta \in (0,1)^2$,  let $u \in H^1_\theta (Q)$ be the unique solution to
\begin{equation}
\label{eq:2dgapprop1}
- \Delta u = \chi_0.
\end{equation}
Then 
\begin{equation}
\label{eq:2dgapprop2}
\int_{B_\delta (0)} u =  \frac{\pi }{8}\delta^4 - \pi^2 \delta^4 \left( \frac{1}{2 \pi} \ln{\delta} - g_0(\theta) \right) ,
\end{equation}
for some constant $g_0(\theta)$ depending on $\theta$.
\end{prop}

\begin{prop}
\label{prop:pcspectlim2}
Let $g_0(\theta)$ be given by Proposition \ref{prop:pcspectlim1}. Then $g_0(\theta)$ is uniformly bounded from above with respect to $\theta$, i.e. there exists a constant $A \in \mathbb{R}$ independent of $\theta$ such that
$$
g_0(\theta) \ge A, \quad \forall \theta \in (0 , 1)^2.
$$ 
\end{prop}

\begin{cor}
\label{cor:pclimitspec}
There exists a $\delta_0$ and a constant $C > 0$ such that for all $\delta < \delta_0$, and all $\theta \in (0, 1)^2$
$$
\int_{B_\delta(0)} u \ge - C \delta^4 \ln{\delta}.
$$
\end{cor}

\begin{proof}[Proof of Proposition  \eqref{prop:pcspectlim1}]
Denote by $G^\theta(x)$ the $\theta$-quasi periodic Green's function for $- \Delta$ on $Q$ with its singularity at the origin. It is known that, by isolating the singularity and expanding in Fourier series in $\theta$,
\begin{equation}
\label{eq:limprop1.1}
G^\theta(y) = -\frac{1}{2\pi} \ln{r} + g_0(\theta) + \sum^\infty_{n=1} g_n^{\pm} r^n e^{ {\rm i} n \varphi}, \quad \text{for $y \in B_\delta(0)$}.
\end{equation}
Here $ g_0(\theta)$ and $g^\pm_n$ are known constants, $r,\varphi$ are the polar coordinates. We shall now construct an explicit solution to \eqref{eq:2dgapprop1}. Define $u$ as follows:
\begin{equation}
\label{eq:limprop1.2}
u = \left\{ \begin{array}{lr} \pi \delta^2 G^\theta(x), & \text{in $Q \backslash B_\delta(0)$} \\ - \frac{r^2}{4} + B + \pi \delta^2 \left(  G^\theta(x) + \frac{1}{2 \pi} \ln{r} - g_0(\theta) \right), & \text{in  $B_\delta(0)$.} \end{array} \right.
\end{equation}
The constant $B$ is chosen such that $u \in H^1_\theta(Q)$, i.e.
$$
B = \frac{\delta^2}{4} - \pi \delta^2 \left( \frac{1}{2 \pi} \ln{\delta} - g_0(\theta) \right).
$$
We will now show that $u$ solves \eqref{eq:2dgapprop1}. For fixed $\varphi \in H^1_{\#}(Q)$
\begin{equation}
\label{eq:limprop1.3}
\int_{Q} \nabla u \cdot \nabla \varphi  = - \int_{Q \backslash B_\delta(0)}  \Delta u \varphi -\int_{ B_\delta(0)}  \Delta u \varphi  + \int_{\partial{B_\delta(0)}} \left[ \pd{u}{n}  \right] \varphi,
\end{equation}
where $n$ is unit outward normal to $B_\delta(0)$, $\pd{u}{n} := \nabla u \cdot n$ and $\left[ \pd{u}{n}  \right]$ is the jump across the interface $\partial{B_\delta(0)}$. From \eqref{eq:limprop1.1} and \eqref{eq:limprop1.2} we notice
\begin{gather}
\label{eq:limprop1.4}
- \int_{Q \backslash B_\delta(0)}  \Delta u \ \varphi  = 0,\\ \nonumber \\
\label{eq:limprop1.5}
-\int_{ B_\delta(0)}  \Delta u \ \varphi  = - \int_{ B_\delta(0)} \Delta( - \frac{r^2}{4}) \varphi = \int_{ B_\delta(0)}\varphi = \int_{ Q} \chi_0 \ \varphi, \\ \nonumber \\
 \left[ \pd{u}{n}  \right]  = \left. \pd{}{r}\left( \pi \delta^2 G^\theta -\left[ - \frac{r^2}{4} + B + \pi \delta^2 \left( G^\theta + \frac{1}{2 \pi} \ln{r} -g_0(\theta) \right) \right] \right) \right\vert_{r=\delta} \nonumber \\  = \left. \pd{}{r} \left[ \frac{r^2}{4} -  \frac{\delta^2}{2} \ln{r} \right]\right\vert_{r=\delta} = 0. \label{eq:limprop1.6}
\end{gather}
Equations \eqref{eq:limprop1.4}-\eqref{eq:limprop1.6} and \eqref{eq:limprop1.3} imply that $u$  solves \eqref{eq:2dgapprop1}. It remains to show \eqref{eq:2dgapprop2}:
\begin{align*}
\int_{B_\delta(0)} u &= \int_{B_\delta(0)} \left[ - \frac{r^2}{4} + B + \pi \delta^2 \left(  G^\theta(x) + \frac{1}{2 \pi} \ln{r} - g_0(\theta) \right) \right] \\ 
& = \int_{B_\delta(0)} \left[ -\frac{r^2}{4} + B \right] + \pi \delta^2 \underbrace{\int_{B_\delta (0)} \sum_{n=1}^\infty g_n^\pm r^n e^{ {\rm i} n \varphi}}_{ = 0} \\
\\ & = - \frac{\pi}{8}\delta^4 + \pi \delta^2 B  = \frac{\pi \delta^4}{8} - \pi^2 \delta^4 \left( \frac{1}{2 \pi} \ln{\delta} - g_0(\theta) \right).
\end{align*} 
\end{proof}

\begin{proof}[Proof of Proposition \eqref{prop:pcspectlim2}]
Consider $G(\theta;k;x)) \in H^1_{\theta}(Q \backslash \left\{ 0 \right\})$ such that 
$$
G(\theta;k;x) = - \frac{1}{2\pi} \ln{r} + g_0(\theta,k) + o(r\ln{r}) \quad \text{as $r \rightarrow 0$},
$$
and
\begin{equation}
\label{eq:limprop2.1}
\left( - \Delta -k \right)G = 0, \quad \text{in $Q \backslash \left\{ 0 \right\} $}.
\end{equation}
This function is well defined, for example, if $k < \vert \theta \vert^2$, which could be seen from explicit analysis of the eigenvalues of $- \Delta$. We aim to show that there exists a constant $A$ such that
\begin{equation}
\label{exlimjunk}
g_0(\theta,0) \ge A, \quad \forall \theta \in [0,1)^2,
\end{equation}
since $G(\theta ; 0 ; y)$ is nothing more than $G^\theta(y)$, the Green's function used in the proof of Proposition \ref{prop:pcspectlim1}.  Note that, for fixed negative $k$, e.g. $k= -1$, $g_0(\theta ; -1)$ is continuous with respect to $\theta \in [ 0 , 1)^2$ and therefore is uniformly bounded from below by some constant $A$, i.e.
$$
g_0(\theta,-1) \ge A, \quad \forall \theta \in [0, 1)^2,
$$
Moreover, $g_0(\theta ; k)$ is continuous with respect to $k \in [-1,0]$. Therefore, to prove \eqref{exlimjunk} it is sufficient to prove
\begin{equation}
\label{eq:limprop2.4}
g'_0(\theta;k) > 0, \quad k \in [-1,0),
\end{equation}
where the prime, $'$, denotes differentiation with respect to $k$. 

We shall now prove \eqref{eq:limprop2.4}. Differentiating \eqref{eq:limprop2.1} with respect to $k$ gives
\begin{gather}
\left( - \Delta -k \right)G^{'} - G = 0, \quad \text{in $Q \backslash\{ 0 \}$} \label{eq:limprop2.2}\\
\intertext{and} 
G^{'}(\theta;k;x) =  g'_0(\theta,k) + \ldots \quad \text{as $r \rightarrow 0$}. \label{eq:limprop2.3}
\end{gather} 
Multiplying \eqref{eq:limprop2.2} by $\overline{G}$ and integrating over $Q \backslash B_\delta(0)$, for sufficiently small $\delta$, gives
\begin{align*}
\int_{Q \backslash B_\delta(0)} G \overline{G} & = \int_{Q \backslash B_\delta(0)} \left( - \Delta -k \right)G^{'} \overline{G} 
=  \int_{\partial{B_\delta(0)}} \left( \pd{}{n} G^{'} \overline{G} - G^{'} \pd{}{n}\overline{G} \right) \\
 & =  \int_{\partial{B_\delta(0)}} \left[ \left( O(1) + \ldots \right) \left( - \frac{1}{2\pi} \ln{r} + \ldots \right) - \left( g'_0(\theta,k) + \ldots \right)\left( - \frac{1}{2\pi r} + \ldots \right) \right] \\
 & =  \int_{\partial{B_\delta(0)}} g'_0 (\theta,k) \frac{1}{2 \pi r} + o(\delta \ln{\delta}),
\end{align*}
passing to the limit $\delta \rightarrow 0$ gives
$$
\int_{Q} \vert G \vert^2 =  g'_0(\theta,k).
$$
This proves \eqref{eq:limprop2.4}.
\end{proof}


\section{Appendix}
Here we formulate a generalisation of the Weyls decomposition,  first introduced in \cite{VPSIVK}, and prove the tensor  $a^{(1)}$ given in matrix representation form as
\begin{align*}
a^{(1)}(y) & = \chi_1(y) \left( \begin{matrix}
1 & 0 & 0 & 1 \\ 0 & 1 & -1 & 0 \\ 0 & -1 & 1 & 0 \\ 1 & 0 & 0 & 1
\end{matrix}
\right),
\end{align*}
satisfies the associated assumptions of the decomposition theorem.
Recalling the spaces
\begin{gather*}
V = \left\{ v \in [H^1_{\#}(Q)]^2 : \text{ $\div{v} = 0$ and  $\div{v^{\perp}} = 0$ in  $Q_1$}\right\}, \\[5pt]
W\,:=\,\left\{\, \psi\,\in\,\left[L^2(Q)\right]^{2\times
2}\,\left\vert \, \mbox{
div}\left(\,\left(a^{(1)}\right)^{1/2}\psi\,\right)\,=\,0
\ \mbox{ in  }
\left[H_\#^{-1}(Q)\right]^2\,\right.\right\},
\end{gather*}
the Weyl's decomposition, associated with $a^{(1)}$, (see \cite{VPSIVK,mythesis} for the proof) is as follows
\begin{lem}[Generalised Weyl's decomposition.]
\label{lem:pdhom3}
Assume that there exists a constant $C>0$ such that for any given $u\in \left[H_\#^1(Q)\right]^2$
there exists $v\in V$ with
\begin{equation}
\norm{u - v}_{[H^1_\#(Q)]^2}\,\leq\, C\,\norm{a^{(1)} \nabla u}_{[L^2(Q)]^{2\times2}}.
\label{eq:pdhom14}
\end{equation} 
\vspace{.1in}
Suppose $\eta  \in \left[L^2(Q)\right]^{2\times 2}$ is orthogonal in $\left[L^2(Q)\right]^{2\times 2}$ to $W$,  i.e.
\begin{equation}
\sum_{i,j=1}^2\int_Q
\eta_{ij}\psi_{ij} =\,0, \ \ \ \forall \,
\psi\,\in\,W. \label{eq:pdhom16}
\end{equation}
Then, there exists $u_1\in\,\left[H_\#^{1}(Q)\right]^2$ such that
\begin{equation}
\eta\,=\,\left(a^{(1)}\right)^{1/2}\nabla u_1.
\label{eq:pdhom17}
\end{equation}
Such a $u_1$ is determined  uniquely up to any function from $V$, in particular $u_1$ is unique in $V^\bot$ the orthogonal complement of $V$ in $[H^1(Q)]^2$.
\end{lem}

We shall now demonstrate that $a^{(1)}$ does indeed satisfy the assumptions of Lemma \ref{lem:pdhom3}. As
$$
\norm{a^{(1)} \nabla v}_{[L^2(Q)]^{2\times2}} = \norm{\div{u}}_{L^2(Q_1)}^2 + \norm{\div{u^{\perp}}}_{L^2(Q_1)}^2,
$$
 we prove the following result.
\begin{lem}
\label{emcoercivelem}
There exists a constant $C>0$ such that for any given $u \in [H^1_{\#}(Q)]^2$ there exists $v\in V$ with
\begin{equation}
\label{emcoercive}
\norm{u - v}_{[H^1_\#(Q)]^2}^2 \le C \left( \norm{\div{u}}_{L^2(Q_1)}^2 + \norm{\div{u^{\perp}}}_{L^2(Q_1)}^2 \right).
\end{equation}
\end{lem}

\begin{proof}Suppose $u \in [H^1_{\#}(Q)]^2$. Let us construct a sufficient function $v \in V$ such that the difference $\norm{u - v}_{[H^1(Q)]^2}^2$ is bounded by the right hand side of \eqref{emcoercive}.

Let $z : =(w_{1,1} - w_{2,2}, w_{1,2} + w_{2,1})$ for $w_1, w_2 \in H^1_{\#}(Q)$ are the solutions to
\begin{align}
\label{coercivelemproof1}
 \Delta w_1 &=  \chi_1 \nabla \cdot u - \chi_0 c_1, & - \Delta w_2 & = \chi_1 \nabla^\perp \cdot u - \chi_0 c_2,
\end{align}
such that $\mv{w_1} = \mv{w_2}=0$. Here $\chi_1$ is the characteristic function of $Q_1$ and $c_1,c_2$ are chosen such that the existence of $w_1$ and $w_2$ is guaranteed, i.e.
\begin{align}
\label{coercivelemproof2}
c_1 = \frac{1}{\vert Q_0 \vert} \int_{Q_1} \nabla \cdot u \ \mathrm{d}y, && c_2 = \frac{1}{\vert Q_0 \vert} \int_{Q_1} \nabla^\perp \cdot u \ \mathrm{d}y.
\end{align}
Then, by Lemma \ref{lem:spcom3.1}(ii), $w_1,w_2$ belong to $H^2_{\#}(Q)$ which implies $z \in [H^1_{\#}(Q)]^2$. Furthermore,
\begin{gather*}
\nabla \cdot z = w_{1,11} - w_{2,21} + w_{1,22} + w_{2,12} = \Delta w_1 =  \chi_1 \nabla \cdot u - \chi_0 c_1 , \\
\intertext{and}
\nabla^{\perp} \cdot z =  - w_{1,21} - w_{2,11} + w_{1,12} - w_{2,22} = - \Delta w_1 =  \chi_1 \nabla^{\perp} \cdot u - \chi_0 c_2.
\end{gather*}
Therefore $v := u - z$ belongs to $V$. Using Lemma \ref{lem:spcom3.1}(ii), \eqref{coercivelemproof1},  \eqref{coercivelemproof2} and cauchy-schwarz inequality we find that
\begin{flalign*}
\norm{u - v}_{[H^1(Q)]^2}^2 & = \norm{z}_{[H^1(Q)]^2}^2 \le  \norm{w_1}_{H^2(Q)}^2 + \norm{w_2}_{H^2(Q)}^2  \\
& \le C \left( \norm{\chi_1 \nabla \cdot u - c_1}_{L^2(Q)}^2 + \norm{\chi_1 \nabla^{\perp} \cdot u - c_2}_{L^2(Q)}^2 \right) \\
& = C \left( \int_{Q_1} \vert \nabla \cdot u \vert^2 + c_1^2 \vert Q_0 \vert +  \int_{Q_1} \vert \nabla^{\perp} \cdot u \vert^2 + c_2^2 \vert Q_0 \vert \right) \\
& \le C \left(  \norm{\nabla \cdot u}_{L^2(Q_1)}^2 + \norm{\nabla^{\perp} \cdot u}_{L^2(Q_1)}^2 \right).
\end{flalign*}
\end{proof}


\section{References}
\begin{thebibliography}{25}
\providecommand{\natexlab}[1]{#1}

%
%
%
%
%
%

\bibitem[8]{nguetseng}
G. Nguetseng, {\itshape A general convergence result for a functional related
  to the theory of homogenization.}, SIAM J.Math Anal 20 (1989), pp. 608--623.

\bibitem[9]{allaire}
G. Allaire, {\itshape Homogenisation and two-scale convergence}, SIAM J.Math
  Anal 23 (1992), pp. 1482 -- 1518.

\bibitem[10]{zhikov1}
V. Zhikov, {\itshape On an extension of the method of two-scale convergence and
  its applications}, Sbornik Math. 191 (2000), pp. 973--1014.

\bibitem[11]{zhikov2}
---{}---{}---, {\itshape On spectrum gaps of some divergent elliptic operators
  with periodic coefficients}, St.Petersburg Math. J. 16 (2005), pp. 773--790.

%
%

\bibitem[15]{VPSIVK}
I. Kamotski and V. Smyshlyaev, {\it Two-scale homogenization for a class of partially degenerating PDE systems}, Preprint arXiv:1309.4579v1.

%
%
%
%
%
%
%
%
%
\bibitem{mythesis}
Cooper, S. 2012.
Two-scale homogenisation of partially degenerating PDEs with applications to photonic crystals and elasticity, {\it PhD Thesis,} University of Bath.
\bibitem{ARROW}
Litchinitser, N., Dunn, S., Steinvurzel, P., Eggleton, B., White, T., McPhedran, R., and De Sterke, C. {\it Application of an arrow model for designing tunable photonic devices}. Opt. Express 12, (8) (2004), 1540–1550.
\end{thebibliography}
\end{document}